\documentclass[11pt, a4paper]{article}
\usepackage[english]{babel}
\usepackage[utf8]{inputenc}
\usepackage[round]{natbib}
\usepackage{scrextend}
\usepackage[a4paper, includeheadfoot, left=2.20cm, right=2.20cm, top=1.5cm, bottom=2.0cm]{geometry} 
\usepackage{setspace}

\usepackage{bbm}
\usepackage{enumitem}
\usepackage{amsmath,amsfonts,amsthm, amssymb} 

\usepackage[font = small]{caption}
\usepackage{graphicx}
\usepackage{tabularx}
\usepackage{caption}
\usepackage{booktabs}
\usepackage{multirow}
\usepackage{lscape}
\usepackage{subcaption}
\usepackage{placeins}
\usepackage{dcolumn}
\allowdisplaybreaks
\usepackage{xcolor}

\usepackage[colorlinks=true, urlcolor=black, linkcolor=black, citecolor=black]{hyperref} 
\newcolumntype{Y}{>{\centering\arraybackslash}X}

\newcommand{\norm}[1]{\left\lVert#1\right\rVert}

\DeclareMathOperator*{\argmax}{arg\,max}
\DeclareMathOperator*{\argmin}{arg\,min}
\newtheorem{proposition}{Proposition}
\newtheorem{lemma}{Lemma}

\deffootnote{0em}{1.6em}{\thefootnotemark.\enskip}

\title{Large-Scale Portfolio Allocation Under Transaction Costs and Model Uncertainty\thanks{ Nikolaus Hautsch (\href{mailto:nikolaus.hautsch@univie.ac.at}{\nolinkurl{nikolaus.hautsch@univie.ac.at}}), University of Vienna, Research Platform ''Data Science @ Uni Vienna'' as well as Vienna Graduate School of Finance (VGSF) and Center for Financial Studies (CFS), Frankfurt. 
Department of Statistics and Operations Research, Faculty of Business, Economics and Statistics, Oskar-Morgenstern-Platz 1, University of Vienna, A-1090 Vienna, Austria, Phone: +43-1-4277-38680, Fax: +43-4277-8-38680.
Stefan Voigt, WU (Vienna University of Economics and Business) and VGSF.
}}
\author{
\begin{tabularx}{0.9\textwidth}{YY}
	Nikolaus Hautsch                             & Stefan Voigt                                
\end{tabularx}
}
\date{\vspace{-1.1cm} }
\begin{document}
\clearpage
\maketitle
\thispagestyle{empty}
\begin{abstract}
	\noindent We theoretically and empirically study portfolio optimization under transaction costs and establish a link between turnover penalization and covariance shrinkage with the penalization governed by transaction costs. We show how the ex ante incorporation of transaction costs shifts optimal portfolios towards regularized versions of efficient allocations.  The regulatory effect of transaction costs is studied in an econometric setting incorporating parameter uncertainty and optimally combining predictive distributions resulting from high-frequency and low-frequency data. In an extensive empirical study, we illustrate that turnover penalization is more effective than commonly employed shrinkage methods and is crucial in order to construct empirically well-performing portfolios.
			
\vspace{1em}


\end{abstract}

\newpage
\clearpage
\setcounter{page}{1}
\section{Introduction}\label{sec:introduction}
Optimizing large-scale portfolio allocations remains a challenge for econometricians and practitioners due to (i) the noisiness of parameter estimates in large dimensions, (ii) model uncertainty and time variations in individual models' forecasting performance, and (iii) the presence of transaction costs, making otherwise optimal rebalancing costly and thus sub-optimal. 

Although there is a huge literature on the statistics of portfolio allocation, it is widely fragmented and typically only focuses on partial aspects. 
For instance, a substantial part of the literature concentrates on the problem of estimating vast-dimensional covariance matrices by means of regularization techniques, see, e.g., \cite{Ledoit.2003, Ledoit.2004, Ledoit.2012} and \cite{Fan.2008}, among others.
This literature has been boosted by the availability of high-frequency (HF) data, which opens an additional channel to increase the precision of covariance estimates and forecasts, see, e.g.,  \cite{BarndorffNielsen.2004}.
Another segment of the literature studies the effects of ignoring parameter uncertainty and model uncertainty arising from changing market regimes and structural breaks.\footnote{The effect of ignoring estimation uncertainty is considered, among others, by \cite{Brown.1976}, \cite{Jobson.1980}, \cite{Jorion.1986} and \cite{Chopra.1993}. 
Model uncertainty is investigated, for instance, by \cite{Wang.2005}, \cite{Garlappi.2007} and \cite{Pflug.2012}.} 
Further literature is devoted to the role of transaction costs in portfolio allocation strategies. 
In the presence of transaction costs, the benefits of reallocating wealth may be smaller than the costs associated with turnover. 
This aspect has been investigated theoretically, among others, for one risky asset by \cite{Magill.1976} and \cite{Davis.1990}. Subsequent extensions to the case with multiple assets have been proposed by \cite{Taksar.1988}, \cite{Akian.1996}, \cite{Leland.1999} and \cite{Balduzzi.2000}.
More recent papers on empirical approaches which explicitly account for transaction costs include \cite{Liu.2004}, \cite{Lynch.2010}, \cite{Garleanu.2013} and \cite{DeMiguel.2014, DeMiguel.2015}. 

Our paper connects the work on shrinkage estimation and transaction costs in two ways: 
First, we show a close connection between covariance regularization and the effects of transaction costs in the optimization procedure.\footnote{The work of \cite{DeMiguel.2018} is close in spirit to our approach. 
Though both papers derive optimal portfolios after imposing a $L_p$ penalty on rebalancing, the implications, however, are different. 
While we focus on the regularization effect, \cite{DeMiguel.2018} point out the close relationship to a robust Bayesian decision problem where the investor imposes priors on her optimal portfolio.}
Second, we empirically document these effects in a large-scale study mimicking portfolio optimization under preferably realistic conditions based on a large panel of assets through more than 10 years.

In fact, in most empirical studies, transaction costs are incorporated \emph{ex post} by analyzing to which extent a certain portfolio strategy would have survived in the presence of transaction costs of given size.\footnote{See, e.g., \cite{DeMiguel.2009b} or \cite{Hautsch.2015}.} 
In financial practice, however, the costs of portfolio rebalancing are taken into account \emph{ex ante} and thus are part of the optimization problem. 
Thus, our objective is to understand the effect of turnover penalization on the ultimate object of interest of the investor -- the optimal portfolio allocation. This focus is notably different from the aim of providing sensible estimates of asset return covariances (via regularization methods), which are then plugged into a portfolio problem. Instead, we show how the presence of transaction costs changes the optimal portfolio and provide an alternative interpretation in terms of parameter shrinkage.
In particular, we illustrate that \emph{quadratic} transaction costs can be interpreted as shrinkage of the variance-covariance matrix towards a diagonal matrix and a shift of the mean that is proportional to transaction costs and current holdings. 
Transaction costs \emph{proportional} to the amount of rebalancing imply a regularization of the covariance matrix, acting similarly as the least absolute shrinkage and selection operator (Lasso) by \cite{Tibshirani.1996} in a regression problem and imply to put more weight on a buy-and-hold strategy. 
The regulatory effect of transaction costs results in better conditioned covariance estimates and  significantly reduces the amount (and frequency) of rebalancing. 
These mechanisms imply strong improvements of portfolio allocations in terms of expected utility and Sharpe ratios compared to the case where transaction costs are neglected.

We perform a reality check by empirically analyzing the role of transaction costs in a high-dimensional and preferably realistic setting. 
We take the perspective of an investor who is monitoring the portfolio allocation on a daily basis while accounting for the (expected) costs of rebalancing. 
The underlying portfolio optimization setting accounts for parameter uncertainty and model uncertainty, while utilizing not only predictions of the covariance structure but also of higher-order moments of the asset return distribution. 
Model uncertainty is taken into account by considering time-varying combinations of predictive distributions resulting from competing models using optimal prediction pooling according to \cite{Geweke.2011b}. This makes the setup sufficiently flexible to utilize a long sample covering high-volatility and low-volatility periods subject to obvious structural breaks. As a by-product we provide insights into the time-varying nature of the predictive ability of individual models under transaction costs and to which extent suitable forecast combinations may result into better portfolio allocations.

The downside of such generality is that the underlying optimization problem cannot be solved in closed form and requires (high-dimensional) numerical integration.
We therefore pose the econometric model in a Bayesian framework, which allows to integrate out parameter uncertainty and to construct posterior predictive asset return distributions based on time-varying mixtures making use of Bayesian computation techniques. 
Optimality of the portfolio weights is ensured with respect to the predicted out-of-sample utility net of transaction costs. 
The entire setup is complemented by a portfolio bootstrap by performing the analysis based on randomized sub-samples out of the underlying asset universe. 
In this way, we are able to gain insights into the statistical significance of various portfolio performance measures.

We analyze a large-scale setting based on all constituents of the S\&P500 index, which are continuously traded on Nasdaq between 2007 and 2017, corresponding to 308 stocks. 
Forecasts of the daily asset return distribution are produced based on three major model classes. 
On the one hand, utilizing HF message data, we compute estimates of daily asset return covariance matrices using blocked realized kernels according to \cite{Hautsch.2012}. 
The kernel estimates are equipped with a Gaussian-inverse Wishart mixture close to the spirit of \cite{Jin.2013} to capture the entire return distribution. 
Moreover, we compute predictive distributions resulting from a daily multivariate stochastic volatility factor model in the spirit of \cite{Chib.2006}.\footnote{So far, stochastic volatility models have been shown to be beneficial in portfolio allocation by \cite{Aguilar.2000} and \cite{Han.2006} for just up to 20 assets.} 
As a third model class, representing traditional estimators based on rolling windows, we utilize the sample covariance and the (linear) shrinkage estimator proposed by \cite{Ledoit.2003, Ledoit.2004}. 

To our best knowledge, this paper provides the first study evaluating the predictive power of high-frequency and low-frequency models in a large-scale portfolio framework under such generality, while utilizing data of $2,409$ trading days and more than 73 billion high-frequency observations. 
Our approach brings together concepts from (i) Bayesian estimation for portfolio optimization, ii) regularization and turnover penalization, (iii) predictive model combinations in high dimensions and (iv) HF-based covariance modeling and prediction. 

We can summarize the following findings: 
First, none of the underlying predictive models is able to produce positive Sharpe ratios when transaction costs are \emph{not} taken into account ex ante. 
This is mainly due to high turnover implied by (too) frequent rebalancing. 
This result changes drastically when transaction costs are considered in the optimization \textit{ex ante}.
Second, when incorporating transaction costs into the optimization problem, performance differences between competing predictive models for the return distribution become smaller. 
It is shown that none of the underlying approaches does produce significant utility gains on top of each other. 
We thus conclude that the respective pros and cons of the individual models in terms of efficiency, predictive accuracy and stability of covariance estimates are leveled out under turnover regularization. 
Third, despite of a similar performance of individual predictive models, mixing high-frequency and low-frequency information is beneficial and yields significantly higher Sharpe ratios. 
This is due to time variations in the individual model's predictive ability. 
Fourth, naive strategies, variants of minimum variance allocations and further competing strategies are statistically and economically significantly outperformed. 

The structure of this paper is as follows: 
Section \ref{sec:setup} theoretically studies the effect of transaction costs on the optimal portfolio structure. 
Section \ref{sec:combinations} gives the econometric setup accounting for parameter and model uncertainty. Section \ref{sec:models} presents the underlying predictive models. 
In Section \ref{sec:empirics}, we describe the data and present the empirical results. 
Finally, Section \ref{sec:conclusion} concludes. 
All proofs, more detailed information on the implementation of the estimation procedures described in the paper as well as additional results are provided in an accompanying Online Appendix. 

\section{The Role of Transaction Costs}\label{sec:setup}

\subsection{Decision Framework}
We consider an investor equipped with a utility function $U_\gamma(r)$ depending on returns $r$ and risk aversion parameter $\gamma$.
At every period $t$, the investor allocates her wealth among $N$ distinct risky assets with the aim to maximize expected utility at $t+1$ by choosing the allocation vector $\omega_{t+1}\in\mathbb{R}^N$. We impose the constraint $\sum\limits_{i=1}^N \omega_{t+1,i} =1$. 
The choice of $\omega_{t+1}$ is based on drawing inference from observed data. 
The information set at time $t$ consists of the time series of past returns ${R_t=\left(r'_{1},\ldots,r'_{t}\right)'\in\mathbb{R}^{t \times N}}$, where $r_{t, i}  :=\frac{p_{t, i}-p_{t-1, i}}{p_{t-1, i}}$ are the simple (net) returns computed using end-of-day asset prices $p_t := (p_{t,1},\ldots,p_{t,N})\in\mathbb{R}^N$. 
The set of information may contain additional variables $\mathcal{H}_t$, as, e.g., intra-day data. 

We define an optimal portfolio as the allocation, which maximizes expected utility of the investor after subtracting transaction costs arising from rebalancing. 
We denote transaction costs as $\nu_t (\omega)$, depending on the desired portfolio weight $\omega$ and reflecting broker fees and implementation shortfall. 

Transaction costs are a function of the distance vector $\Delta_t := \omega_{t+1} - \omega_{t^+}$ between the new allocation $\omega_{t+1}$ and the allocation right before readjustment,  $\omega_{t^+} := \frac{\omega_t \circ \left(1+r_{t}\right)}{1 + \omega_t ' r_{t}}$, where the operator $\circ$ denotes element-wise multiplication. The vector $\omega_{t^+}$ builds on the allocation $\omega_t$, which has been considered as being optimal given expectations in $t-1$, but effectively changed due to returns between $t-1$ and $t$. 

At time $t$, the investor monitors her portfolio and solves a static maximization problem conditional on her current beliefs on the distribution of returns $p_t\left(r_{t+1}|\mathcal{D}\right) := p(r_{t+1}|R_t,\mathcal{H}_t, \mathcal{D})$ of the next period and the current portfolio weights $\omega_{t^+}$:
\begin{align*}\label{al:EU}\omega_{t+1} ^* &:=\argmax_{\omega\in\mathbb{R}^N,\text{ }  \iota'\omega=1} E\left(U_\gamma\left(\omega ' \left(1+r_{t+1}\right)-\nu_{t}(\omega)\right)|R_t,\mathcal{H}_t, \mathcal{D}\right) \\\tag{EU} &= \argmax_{\omega\in\mathbb{R}^N,\text{ } \iota'\omega=1} \int\limits_{\mathbb{R}^N} U_\gamma\left(\omega '\left(1+r_{t+1}\right)-\nu_t (\omega)\right)p_t (r_{t+1}|\mathcal{D})\text{d}r_{t+1}\end{align*}
where $\iota$ is a vector of ones. 
Note that optimization problem (\ref{al:EU}) reflects the problem of an investor who constantly monitors her portfolio and exploits all available information, but rebalances only if the costs implied by deviations from the path of optimal allocations exceed the costs of rebalancing.
This form of myopic portfolio optimization ensures optimality (after transaction costs) of allocations at each point in time.\footnote{The investment decision is myopic in the sense that it does not take into account the opportunity to rebalance the portfolio next day. 
Thus, hedging demand for assets is not adjusted for. 
Although attempts exist to investigate the effect of (long-term) strategic decisions of the investor, numerical solutions are infeasible in our high-dimensional asset space (see, e.g. \cite{Campbell.2003}).
Further, under certain restrictions regarding the utility function of the investor, the myopic portfolio choice indeed represents the optimal solution to the (dynamic) control problem.} 
Accordingly, the optimal wealth allocation $\omega_{t+1} ^*$ from representation (\ref{al:EU}) is governed by i) the structure of turnover penalization $\nu_t(\omega)$, and ii) the return forecasts $p_t (r_{t+1}|\mathcal{D})$. 
 
\subsection{Transaction Costs in Case of Gaussian Returns}
In general, the solution to the optimization problem (\ref{al:EU}) cannot be derived analytically but needs to be approximated using numerical methods. However, assuming $p _t(r_{t+1}|\mathcal{D})$ being a multivariate normal density with known parameters $\Sigma$ and $\mu$, problem (\ref{al:EU}) coincides with the initial \cite{Markowitz.1952} approach and yields an analytical solution, resulting from the maximization of the certainty equivalent (CE) after transaction costs,
\begin{align}
\label{equ:CE_representation}\omega_{t+1} ^* & = \argmax _{\omega \in \mathbb{R}^N,  \iota'\omega = 1} \omega'\mu - \nu_t \left(\omega\right) - \frac{\gamma}{2}\omega'\Sigma\omega.                        
\end{align}
\subsubsection{Quadratic transaction costs}
We model the transaction costs $
\nu_t\left(\omega_{t+1} \right):=\nu\left(\omega_{t+1},\omega_{t^+}, \beta \right)$ for shifting wealth from allocation $\omega_{t^+}$ to $\omega_{t+1}$ as a quadratic function given by \begin{align}\nu_{L_2}\left(\omega_{t+1},\omega_{t^+}, \beta\right) := \frac{\beta}{2} \left(\omega_{t+1} - \omega_{t^+}\right)'\left(\omega_{t+1}- \omega_{t^+}\right)
\end{align}
with cost parameter $\beta>0$. 
The allocation $\omega_{t+1} ^*$ according to (\ref{equ:CE_representation}) can then be restated as
\begin{align}
	\label{equ:maximization_naive}\omega_{t+1} ^* = \argmax _{\omega \in \mathbb{R}^N,  \iota'\omega = 1} \omega'\mu - \frac{\beta}{2} \left(\omega - \omega_{t^+}\right)'\left(\omega- \omega_{t^+}\right) - \frac{\gamma}{2}\omega'\Sigma\omega                               = \argmin_{\omega\in\mathbb{R}^N,\text{ }  \iota'\omega=1} 
\frac{\gamma}{2}\omega'\Sigma^* \omega -\omega'\mu^*, 
\end{align}
with
\begin{align}
	\label{ref:Sigma_star}\Sigma^* & := \frac{\beta}{\gamma} I + \Sigma,\\
	\mu^*    & :=\mu+\beta \omega_{t^+},
\end{align}
where $I$ denotes the identity matrix.
The optimization problem with quadratic transaction costs can be thus interpreted as a classical mean-variance problem \emph{without} transaction costs, where (i) the covariance matrix $\Sigma$ is regularized  towards the identity matrix (with $\frac{\beta}{\gamma}$ serving as shrinkage parameter) and (ii) the mean is shifted by $\beta \omega_{t^+}$. 

The shift from $\mu$ to $\mu^* = \mu + \beta \omega_{t^+}$ can be alternatively interpreted by exploiting $\omega'\iota = 1$ and reformulating the problem as 
\begin{equation}
\omega_{t+1}^* = \argmin_{\omega\in\mathbb{R}^N,\text{ }  \iota'\omega=1} 
\frac{\gamma}{2}\omega'\Sigma^* \omega -\omega'\left(\mu + \beta \left(\omega_{t^+} - \frac{1}{N}\iota\right) \right).
\end{equation}
The shift of the mean vector is therefore proportional to the deviations of the current allocation to the $1/N$-setting. This can be interpreted as putting more weight on assets with (already) high exposure.

Hence, if $\beta$ increases, $\Sigma^*$ is shifted towards a diagonal matrix representing the case of uncorrelated assets. 
The regularization effect of $\Sigma^*$ exhibits some similarity to the implications of the shrinkage approach proposed by \cite{Ledoit.2003, Ledoit.2004}. 
While the latter propose to shrink the eigenvalues of the covariance matrix $\Sigma$,  \eqref{ref:Sigma_star} implies shifting all eigenvalues by a constant $\frac{\beta}{\gamma}$ which has a stabilizing effect on the conditioning of $\Sigma^*$. As shown in Section A.1 of the Online Appendix, a direct link to (linear) shrinkage can be constructed by imposing asset-specific transaction costs of the form $v(\omega_{t+1},\omega_{t^+}, B) := \left(\omega_{t+1}-\omega_{t^+}\right)'B\left(\omega_{t+1}-\omega_{t^+}\right)$ for a positive definite matrix $B$. Then, the linear shrinkage estimator of \cite{Ledoit.2003, Ledoit.2004} originate as soon as transaction costs are specifically related to the eigenvalues of $\Sigma$.

Note, however, that in our context, transaction costs $\beta$ are treated as an exogenous parameter, which is determined by the institutional environment. Thus it is not a 'free' parameter and is independent of $\Sigma$. This makes the transaction-cost-implied regularization different to the case of linear shrinkage and 
creates a seemingly counterintuitive result (at first sight): 
Due to the fact that transaction costs $\beta$ are independent of $\Sigma$, periods of high volatility (i.e., an increase of $\Sigma$) imply a relatively lower weight on the identity matrix $I$ and thus less penalization compared to the case of a period where $\Sigma$ is scaled downwards. 
This is opposite to the effect which is expected for shrinkage, where higher volatility implies higher parameter uncertainty and thus stronger penalization. 
As shown by the lemma below, however, the impression of counterintuition disappears if one turns attention not to the covariance matrix \emph{per se}, but to the objective of portfolio optimization.
\begin{lemma}
	Assume a regime of high volatility with $\Sigma^h = (1+h)\Sigma$, where $h>0$ and $\Sigma$ is the asset return covariance matrix during \textit{calm} periods. Assume $\mu=0$. 
	Then, the optimal weight $\omega_{t+1}^*$ is equivalent to the optimal portfolio based on $\Sigma$ and (reduced) transaction costs $\frac{\beta}{1+h} < \beta$.  
\end{lemma} 
\begin{proof}
	See Section A.2 in the Online Appendix.
\end{proof}
\noindent Lemma 1 implies that during periods with high volatility (and covariances) the optimal portfolio decision is triggered \textit{less} by rebalancing costs, but \textit{more} by the need to reduce portfolio volatility. 
The consequence is a shift towards the (global) minimum-variance portfolio with weights $\omega \propto \left(\frac{\beta}{(1+h)\gamma}I + \Sigma\right)^{-1}\iota$. 
Alternatively, one can associate this scenario with a situation of higher risk aversion and thus a stronger need for optimal risk allocation.

The view of treating $\beta$ as an exogenous transaction cost parameter can be contrasted to recent approaches linking transaction costs to volatility. In fact, if transaction costs are assumed to be \emph{proportional} to volatility, $v\left(\omega_{t+1},\omega_{t^+}, \beta\Sigma\right) = \frac{\beta}{2}  \Delta_t'\Sigma\Delta_t$, as, e.g., suggested by \cite{Garleanu.2013}, we obtain $\Sigma^* = \left(1+\frac{\beta}{\gamma}\right)\Sigma$. As shown in Lemma 2 in Section A.2 in the Online Appendix, in this case,  the optimal weights take the form $\omega^{\beta \Delta}_{t+1} = \omega^{\gamma + \beta}_{t+1} + \frac{\beta}{\beta + \gamma}\left(\omega_{t^+} - w^\text{mvp}\right)$, where $\omega^{\gamma + \beta}_{t+1}$ is the efficient portfolio \textit{without} transaction costs, parameters $\Sigma$ and $\mu$ and risk aversion parameter $\gamma + \beta$, and $\omega^\text{mvp} := \frac{1}{\iota'\Sigma^{-1}\iota}\Sigma^{-1}\iota$ corresponds to the minimum variance allocation. Hence, the optimal weights are a linear combination of $\omega^{\gamma + \beta}_{t+1}$ with weight $1$ and of $(\omega_{t^+} - w^\text{mvp})$ with weight $\frac{\beta}{\beta + \gamma}$. Changes in $\Sigma$ thus only affect $\omega^{\gamma + \beta}_{t+1}$ and $w^\text{mvp}$ but \emph{not} the weight $\frac{\beta}{\beta + \gamma}$, which might be interpreted as a 'shrinkage intensity' (in a wider sense).
Hence, the regulatory effect of \emph{endogenous} transaction costs can be ambiguous and very much depends on the exact specification of the link between $\beta$ and $\Sigma$. In contrast, \emph{exogenous} transaction costs imply a distinct regulatory effect, which is shown to be empirically successful as documented in Section \ref{sec:empirics}. 

Focusing on $v_{L_2}\left(\omega_{t+1},\omega_{t^+},\beta\right)$ allows to us to investigate the relative importance of the current holdings on the optimal financial decision.
Proposition \ref{proposition:limit} shows the effect of rising transaction costs on optimal rebalancing.
\begin{proposition}\label{proposition:limit}
\begin{align}
\lim_{\beta\rightarrow\infty}\omega_{t+1}^* 
&= \left(I - \frac{1}{N}\iota \iota' \right) \omega_{t^+}  + \frac{1}{N}\iota = \omega_{t^+}.
\end{align}
\end{proposition}
\begin{proof}
	See Section A.2 in the Online Appendix.
\end{proof}
\noindent Hence, if the transaction costs are prohibitively large, the investor may not implement the efficient portfolio despite her knowledge of the true return parameters $\mu$ and $\Sigma$.
The effect of transaction costs in the long-run can be analyzed in more depth by considering the well-known representation of the mean-variance efficient portfolio,
\begin{align}
\omega(\mu,\Sigma) & := \frac{1}{\gamma}\left(\underbrace{{\Sigma}^{-1} - \frac{1}{\iota' {\Sigma}^{-1}\iota }{\Sigma}^{-1}\iota\iota' {\Sigma}^{-1} }_{:=A(\Sigma)}\right) \mu  + \frac{1}{\iota' {\Sigma}^{-1} \iota }{\Sigma}^{-1} \iota. 
\end{align}
If $\omega_0$ denotes the initial allocation, sequential rebalancing allows us to study the long-run effect, given by
\begin{align}
\omega_T & = \sum\limits_{i=0}^{T-1}\left( \frac{\beta}{\gamma}A(\Sigma^*)\right) ^i \omega(\mu,\Sigma^*) + \left(\frac{\beta}{\gamma}A(\Sigma^*)\right)^T \omega_0.
\end{align}
Therefore, $\omega_T$ can be interpreted as a weighted average of $\omega(\mu,\Sigma^*)$  and the initial allocation $\omega_0$, where the weights depend on the ratio $\beta/\gamma$. 
Proposition \ref{proposition:convergence_high_beta} states that for $\beta\rightarrow\infty$, the long-run optimal portfolio $\omega_T$ is driven only by the initial allocation $\omega_0$.
\begin{proposition}\label{proposition:convergence_high_beta}
$\lim\limits_{\beta\rightarrow\infty}\omega_T = \omega_0$.
\end{proposition}
\begin{proof}
	See Section A.2 in the Online Appendix.
\end{proof}
\noindent The following proposition shows, however, that a range for $\beta$ exists (with critical upper threshold), for which the initial allocation can be ignored in the long-run.
\begin{proposition}\label{proposition:convergence}
$\exists \beta^*> 0 \text{ } \forall \beta\in [0,\beta^*): \left\|\left(\frac{\beta}{\gamma}A(\Sigma^*)\right)\right\|_F < 1$, where $\|\cdot\|_F$ denotes the Frobenius norm $\|A\|_F := \sqrt{\sum\limits_{i=1}^N\sum\limits_{i=1}^N A_{i,j}^2}$ for an $N \times N$ matrix $A$.
\end{proposition}
\begin{proof}
	See Section A.2 in the Online Appendix.
\end{proof}
\noindent Using Proposition \ref{proposition:convergence} for $T\rightarrow \infty$ and $\beta < \beta^*$, the series $\sum_{i=0}^{T}\left( \frac{\beta}{\gamma}A(\Sigma^*)\right) ^i$ converges to $\left(I - \frac{\beta}{\gamma}A(\Sigma^*)\right)^{-1}$ and $\lim_{i\rightarrow\infty}\left( \frac{\beta}{\gamma}A(\Sigma^*)\right) ^i = 0$.
In Proposition \ref{proposition:banach} we show that in the long run, we obtain convergence towards the efficient portfolio:
\begin{proposition}\label{proposition:banach} For $T\rightarrow \infty$ and $\beta<\beta^*$ the series $\omega_T$ converges to a unique fix-point given by
\begin{equation}
\omega_\infty  = \left(I - \frac{\beta}{\gamma}A(\Sigma^*)\right)^{-1} \omega(\mu,\Sigma^*) = \omega\left(\mu,\Sigma\right).
\end{equation}
\end{proposition}
\begin{proof}
	See Section A.2 in the Online Appendix.
\end{proof}
\noindent Note, that the location of the initial portfolio $\omega_0$ itself does not play a role on the upper threshold $\beta^*$ ensuring long-run convergence towards $\omega_{\infty}$. Instead, $\beta^*$ is affected only by the risk aversion $\gamma$ and the eigenvalues of $\Sigma$.
	
\subsubsection{Proportional \texorpdfstring{($L_1$)}{(L1)} transaction costs}
Though attractive from an analytical perspective, transaction costs of quadratic form may represent an unrealistic proxy of costs associated with trading in real financial markets, see, among others, \cite{Toth.2011} and \cite{Robert.2012}. Instead, in the literature, there is widespread use of transaction cost measures proportional to the sum of absolute rebalancing ($L_1$-norm of rebalancing), which impose a stronger penalization on turnover and are more realistic.\footnote{Literature typically employs a penalization terms of 50 bp, see, for instance, \cite{DeMiguel.2009b} and \cite{DeMiguel.2018}. } 
Transaction costs proportional to the $L_1$-norm of rebalancing $\Delta_{t} = \omega_{t+1} - \omega_{t^+}$ yield the form
\begin{align}\label{al:L1costs}
	\nu_{L_1}\left(\omega_{t+1},\omega_{t^+}, \beta\right) := \beta \norm{\omega_{t+1} - \omega_{t^+}}_1 = \beta \sum\limits_{i=1}^N \left|\omega_{t+1, i} - \omega_{t^+, i}\right|,
\end{align} with cost parameter $\beta>0$.
Although the effect of proportional transaction costs on the optimal portfolio cannot be derived in closed-form comparable to the quadratic ($L_2$) case discussed above, the impact of turnover penalization can be still interpreted as a form of regularization. 
If we assume for simplicity of illustration $\mu=0$, the optimization problem (\ref{equ:CE_representation}) corresponds to 
\begin{align}\label{al:EUL1}
\omega_{t+1} ^* &:=\argmin_{\omega\in\mathbb{R}^N,\text{ }  \iota'\omega=1} 
\frac{\gamma}{2}\omega'\Sigma \omega + \beta \norm{\omega_{t+1} - \omega_{t^+}}_1 \\
&= \argmin_{\Delta_t\in\mathbb{R}^N,\text{ }  \iota'\Delta_t=0} 
\gamma\Delta_t'\Sigma\omega_{t^+} + \frac{\gamma}{2}\Delta_t'\Sigma \Delta_t + \beta \norm{\Delta_t}_1.
\end{align}
The first-order conditions of the constrained optimization are
\begin{align}
\gamma\Sigma \underbrace{\left(\Delta_t + \omega_{t^+}\right)}_{\omega_{t+1}^*} + \beta \tilde g - \lambda \iota=0, \\ \Delta_t '\iota=0,
\end{align}
where $\tilde g$ is the vector of sub-derivatives of $\norm{\omega_{t+1} - \omega_{t^+}}_1$, i.e., $\tilde  g:=\partial \norm{\omega_{t+1} - \omega_{t^+}}_1 / \partial \omega_{t+1}$, consisting of elements which are $1$ or $-1$ in case $\omega_{t+1, i} - \omega_{t^+, i}>0$ or $\omega_{t+1, i} - \omega_{t^+, i}<0$, respectively, or $\in [-1,1]$ in case $\omega_{t+1, i} - \omega_{t^+, i}=0$.
Solving for $\omega_{t+1}^*$ yields
\begin{align}
\omega_{t+1}^* =\left(1+\frac{\beta}{\gamma}\iota'\Sigma^{-1} \tilde g\right)\omega^\text{mvp} -\frac{\beta}{\gamma}\Sigma^{-1} \tilde g,
\end{align}
where $\omega^\text{mvp} := \frac{1}{\iota'\Sigma^{-1}\iota}\Sigma^{-1}\iota$ corresponds to the weights of the GMV portfolio. Proposition  \ref{proposition:l1costs} shows that this optimization problem can be formulated as a standard  (regularized) minimum variance problem.

\begin{proposition}\label{proposition:l1costs}
	Portfolio optimization problem (\ref{al:EUL1}) is equivalent to the minimum variance problem with 
	\begin{align}
	\omega_{t+1} ^* = \argmin_{\omega\in\mathbb{R}^N,\text{ }  \iota'\omega=1} \omega' \Sigma_{\frac{\beta}{\gamma}} \omega,
	\end{align}
	where $\Sigma_{\frac{\beta}{\gamma}} = \Sigma + \frac{\beta}{\gamma}\left(g^*\iota' + \iota g^{* '}\right)
	$, and $g^*$ is the subgradient of $\norm{\omega^*_{t+1}-\omega_{t^+}}_1$.
\end{proposition}
\begin{proof}
	See Section A.2 in the Online Appendix.
\end{proof}
\noindent The form of the matrix $\Sigma_{\frac{\beta}{\gamma}}$ implies that for high transaction costs $\beta$, more weight is put on those pairs of assets, whose exposure is rebalanced in the same direction. 
The result shows some similarities to \cite{Fan.2012}, who illustrate that the risk minimization problem with constrained weights
\begin{align}\label{equ:fanregularization}
\omega_{t+1} ^{\vartheta}&=\argmin_{\omega\in\mathbb{R}^N,\text{ }  \iota'\omega=1, \text{ } ||\omega||_1 \le \vartheta} \omega'\Sigma \omega
\end{align}
can be interpreted as the minimum variance problem
\begin{align}\label{equ:fanregularization2}
\omega_{t+1} ^{\vartheta}&=\argmin_{\omega\in\mathbb{R}^N,\text{ }  \iota'\omega=1} \omega'\Sigma^{\vartheta} \omega,
\end{align}
with
$
\Sigma^{\vartheta}: =\Sigma + \lambda \left(g \iota' + \iota g'\right),
$
where $\lambda$ is a Lagrange multiplier and $g$ is the subgradient vector of the function $\norm{\omega}_1$ evaluated at the the solution of optimization problem (\ref{equ:fanregularization}).
Note, however, that in our case, the transaction cost parameter $\beta$ is given to the investor, whereas $\vartheta$ is an endogenously imposed restriction with the aim to decrease the impact of estimation error.

Investigating the long-run effect of the initial portfolio $\omega_0$ in the presence of $L_1$ transaction costs in the spirit of Proposition \ref{proposition:banach} is complex as analytically tractable representations are not easily available.
General insights from the $L_2$ benchmark case, however, can be transferred to the setup with $L_1$ transaction costs: First, a high cost parameter $\beta$ may prevent the investor from implementing the efficient portfolio. 
Second, as the $L_2$ norm of any vector is bounded from above by its $L_1$ norm, $L_1$ penalization is always stronger than in the case of quadratic transaction costs. 
Therefore, we expect that the convergence of portfolios from the initial allocation $\omega_0$ towards the efficient portfolio is generally slower, but qualitatively similar.

\subsubsection{Empirical implications}
To illustrate the effects discussed above, we compute the performance of portfolios after transaction costs based on $N = 308$ assets and daily readjustments based on data ranging from June $2007$ to March $2017$.\footnote{A description of the dataset and the underlying estimators is given in more detail in Section \ref{sec:empirics.data}.} 
The unknown covariance matrix $\Sigma_t$ is estimated in two ways: We compute the sample covariance estimator $S_t$ and the shrinkage estimator $\hat \Sigma ^{t,\text{Shrink}}$ by \cite{Ledoit.2003, Ledoit.2004} with a rolling window of length $h = 500$ days. We refrain from estimating the mean and set $\mu_t=0$. 
The initial portfolios weights are set to $\frac{1}{N}\iota$, corresponding to the naive portfolio.
Then, for a fixed  $\beta$ and daily estimates of $\Sigma_t$, portfolio weights are rebalanced as solutions of optimization problem \eqref{equ:maximization_naive} using $\gamma=4$. This yields a time series of optimal portfolios $\omega_t^\beta$ and realized portfolio returns $r^\beta_t := r_t ^\prime \omega_t^\beta$. Subtracting transaction costs then yields the realized portfolio returns net of transaction costs $r_t^{\beta, \text{nTC}} :=r^\beta_t - \beta \norm{\omega^\beta_{t+1} - \omega^\beta_{t^+}}_2$. 

Figure \ref{fig:sharpe_ratio_varying_beta} displays annualized Sharpe ratios, computed as the ratio of the annualized sample mean and standard deviation of $r_t^{\beta, \text{nTC}}$,  after subtracting transaction costs for both the sample covariance and the shrinkage estimator for a range of different values of $\beta$, measured in basis points. 
As a benchmark, we provide the corresponding statistics for the naive portfolio without any rebalancing (buy and hold) and a naive version which reallocates weights back to $\frac{1}{N}\iota$ on a daily basis. 
We report values of $\beta$ ranging from $1$bp to $1,000,000$bp. 
To provide some guidance for the interpretation of these values, note that, for instance, for the naive allocation with daily readjustment, the adjusted returns after transaction costs are 
\begin{equation}
	r_{t+1}^\text{Naive, nTC} := r_{t+1}^\text{Naive} - \beta\norm{\omega_{t+1} - \omega_{t^+}}_2
	= \bar{r}_{t+1} - \beta \frac{1}{N}\frac{1}{\left(1+ \bar{r}_{t+1} \right)^2}\bar{\sigma}_{t+1}^2,
\end{equation}
where $\bar{r}_t$ is the day $t$ average simple return across all $N$ assets and $\bar{\sigma}_t^2 = 1/N\sum_{i=1}^{N}\left(r_{t,i}-\bar{r}_t\right)^2$.
Using sample estimates based on our data, penalizing turnover such that average daily returns of the naive portfolio are reduced by one basis point (which we consider very strict for the naive allocation), $\beta$ should be at least in the magnitude of $3,000$ bp.\footnote{In our sample ($N=308$), the average daily return $\bar{r}$ is roughly $3bp$. The average variance $\bar{\sigma}^2$ is of the same magnitude.} 
These considerations show that scenarios with $\beta<100$ can be associated with rather small transaction costs.

Nonetheless, we observe that the naive portfolio is clearly outperformed for a wide range of values of $\beta$. 
This is remarkable, as it is well-known that parameter uncertainty especially in high dimensions often leads to superior performance of the naive allocation (see, e.g. \cite{DeMiguel.2009b}). 
We moreover find that already very small values of $\beta$ have a strong regulatory effect on the covariance matrix. Recall that quadratic transaction costs directly affect the diagonal elements and thus the eigenvalues of $\Sigma_t$. 
Inflating each of the $308$ eigenvalues by a small magnitude has a substantial effect on the conditioning of the covariance matrix. 
We observe that transaction costs of just 1 bp significantly increase the conditioning number and strongly stabilize $S_t$ and particularly its inverse, $S_t^{-1}$. 
In fact, the optimal portfolio based on the sample covariance matrix, which adjusts ex-ante for transaction costs of $\beta = 1 bp$, leads to a net-of-transaction-costs Sharpe ratio of $0.51$, whereas neglecting transaction costs in the optimization yields a Sharpe ratio of only $0.21$.
This effect is to a large extent a pure regularization effect. 
For rising values of $\beta$, this effect marginally declines and leads to a declining performance for values of $\beta$ between $10$bp and $50$bp.
We therefore observe a dual role of transaction costs. 
On the one hand, they improve the conditioning of the covariance matrix by inflating the eigenvalues. 
On the other hand, they reduce the mean portfolio return. 
Both effects influence the Sharpe ratio in opposite direction causing the concavity of graphs for values of $\beta$ up to approximately $50$bp. 

For higher values of $\beta$ we observe a similar pattern: Here, a decline in rebalancing costs due to the implied turnover penalization kicks in and implies an increase of the Sharpe ratio. 
If the cost parameter $\beta$, however, gets too high, the adverse effect of transaction costs on the portfolio's expected returns dominate. 
Hence, as predicted by Proposition \ref{proposition:limit}, the allocation is ultimately pushed back to the performance of the buy and hold strategy (with initial naive allocation $1/N$). 
This is reflected by the second panel of Figure \ref{fig:sharpe_ratio_varying_beta}, showing the average $L_1$ distance to the buy and hold portfolio weights.
 
The described effects are much more pronounced for the sample covariance than for its shrunken counterpart. 
As the latter is already regularized, additional regulation implied by turnover penalization obviously has a lower impact.  
Nevertheless, turnover regularization implies that forecasts even based on the sample covariance generate reasonable Sharpe ratios, which tend to perform equally well than those implied by a (linear) shrinkage estimator.
The third panel of Figure \ref{fig:sharpe_ratio_varying_beta} illustrates the average turnover (in percent). Clearly, with increasing $\beta$, trading activity declines as shifting wealth becomes less desirable. 
Hence, we observe that differences in the out-of-sample performance between the two approaches decline if the turnover regularization becomes stronger.

Our motivation for reporting Sharpe ratios in Figure \ref{fig:sharpe_ratio_varying_beta} originates from the fact that the Sharpe ratio is a very common metric for portfolio evaluation in practice. 
Note, however, that the maximization of the Sharpe ratio is \emph{not} guaranteed by our underlying objective function \eqref{al:EU} that maximizes expected utility. 
Recall that in the latter framework, for a given $\beta$, allocation $\omega_{t+1}^*$ is optimal. 
If, however, one leaves this ground and aims at a different objective, such as, e.g., the maximization of the Sharpe ratio, one may be tempted to view $\beta$ as a free parameter which needs to be chosen such that it maximizes the out-of-sample performance in a scenario where turnover with transaction costs $v_t\left(\omega_{t+1},\omega_{t^+},\tilde{\beta}\right)$ is penalized \textit{ex post}. 
Hence, in such a setting it might be useful to distinguish between an 
(endogenously chosen) \textit{ex-ante} parameter $\beta$ which differs from the (exogenously) given transaction cost parameter $\tilde{\beta}$ used in the \textit{ex-post} evaluation. 
In Section A.3 of the Online Appendix we provide some evidence for this effect: Instead of computing the weights using the theoretically optimal penalty factor $\beta$, we compute Sharpe ratios for weights based on \textit{sub-optimal ex-ante} $\tilde\beta$s. 
The results in Table A.1 show that small values for an \textit{ex-ante} $\beta$  are sufficient to generate high Sharpe ratios, even if the \textit{ex-post} $\tilde\beta$ chosen for the evaluation is much higher. 
Furthermore, the table suggests that choosing $\beta=\tilde{\beta}$ does not necessarily provide the strongest empirical performance: 
In our setup, $\beta = 25 bp$ leads to the highest Sharpe ratios even if $\tilde{\beta}$ is much higher. 

Repeating the analysis based on proportional ($L_1$) transaction costs delivers qualitatively similar results.
For low values of $\beta$, the Sharpe ratio increases in $\beta$ as the effects of covariance regularization and a reduction of turnover overcompensate the effect of declining portfolio returns (after transaction costs). 
Overall, the effects, however, are less pronounced than in the case of quadratic transaction costs.\footnote{Note that $L_1$ transaction costs imply a regularization which acts similarly to a Lasso penalization. 
Such a penalization implies a strong dependence of the portfolio weights on the previous day's allocation.
This affects our evaluation as the paths of the portfolio weights may differ substantially over time if the cost parameter $\beta$ is slightly changed. 
A visualization of the performance for the $L_1$ case similarly to Figure \ref{fig:sharpe_ratio_varying_beta} is provided in Section A.4 of the Online Appendix (Figure A.1).}

\section{Basic Econometric Setup}\label{sec:combinations}

The optimization problem \eqref{al:EU} posses the challenge of providing a sensible density $p _t(r_{t+1}|\mathcal{D})$ of future returns. The predictive density should reflect dynamics of the return distribution in a suitable way, which opens many different dimensions on how to choose a model $\mathcal{M}_k$. 
The model $\mathcal{M}_k$ reflects assumptions regarding the return generating process in form of a likelihood function $\mathcal{L}\left(r_t|\Theta, \mathcal{H}_t, \mathcal{M}_k\right)$ depending on unknown parameters $\Theta$. 
Assuming that future returns are distributed as $\mathcal{L}\left(r_t|\hat\Theta, \mathcal{H}_t, \mathcal{M}_k\right)$, where $\hat \Theta$ is a point estimate of the parameters $\Theta$, however, would imply that the uncertainty perceived by the investor ignores estimation error, see, e.g. \cite{Kan.2007}. 
Consequently, the resulting portfolio weights would be sub-optimal.

To accommodate parameter uncertainty and to impose a setting, where the optimization problem \eqref{al:EU} can be naturally addressed by numerical integration techniques, we employ a Bayesian approach. 
Hence, by defining a model $\mathcal{M}_k$ implying the likelihood $\mathcal{L}\left(r_t|\Theta, \mathcal{H}_t, \mathcal{M}_k\right)$ and choosing a prior distribution $\pi(\Theta)$,
the posterior distribution 
\begin{equation}
\pi(\Theta|R_{t},\mathcal{H}_t, \mathcal{M}_k)\propto\mathcal{L}(R_t|\Theta, \mathcal{H}_t, \mathcal{M}_k)\pi(\Theta)
\end{equation}
reflects beliefs about the distribution of the parameters after observing the set of available information, $\left(R_t, \mathcal{H}_t\right)$.
The (posterior) predictive distribution of the returns is then given by
\begin{align}\label{equ:posteriorpredictive}
r_{\mathcal{M}_k, t+1} \sim p(r_{t+1}|R_t, \mathcal{H}_t, \mathcal{M}_k):=\int \mathcal{L}(r_{t+1}|\Theta, \mathcal{H}_t, \mathcal{M}_k)\pi(\Theta|R_{t},\mathcal{H}_t, \mathcal{M}_k) d\Theta.
\end{align}
If the precision of the parameters estimates is low, the posterior distribution $\pi(\Theta|R_{t},\mathcal{H}_t, \mathcal{M}_k)$ yields a predictive return distribution with more mass in the tails than focusing only on $\mathcal{L}r_{t+1}|\hat \Theta, \mathcal{H}_t, \mathcal{M}_k$. 

Moreover, potential structural changes in the return distribution and time-varying parameters make it hard to identify a single predictive model which consistently outperforms all other models. 
Therefore, an investor may instead combine predictions of $K$ distinct predictive models $\mathcal{D}:=\left\{\mathcal{M}_1,\ldots,\mathcal{M}_K\right\}$, reflecting either personal preferences, data availability or theoretical considerations.\footnote{
Model combination in the context of return predictions has a long tradition in econometrics, starting from \cite{Bates.1969}. See also \cite{Stock.1999, Stock.2002, Weigend.2000, Bao.2007} and \cite{Hall.2007}. 
In finance, \cite{Avramov.2003} and \cite{Uppal.2003}, among others, apply model combinations and investigate the effect of model uncertainty on financial decisions.}
Stacking the predictive distributions yields
\begin{equation}
r^\text{vec}_{\mathcal{D},t+1} := \text{vec}\left(\left\{r_{\mathcal{M}_1,t+1}, \ldots,r_{\mathcal{M}_K,t+1}\right\}\right) \in\mathbb{R}^{NK \times 1}.
\end{equation} 
The joint predictive distribution $p_t(r_{t+1}|\mathcal{D})$ is computed conditionally on combination weights $c_t\in\mathbb{R}^K$, which can be interpreted as discrete probabilities over the set of models $\mathcal{D}$. The probabilistic interpretation of the combination scheme is justified by enforcing that all weights take positive values and add up to one,
\begin{equation}
c_t\in\Delta_{[0,1]^K}:=\left\{c\in \mathbb{R}^K: c_i\geq 0 \text{ }\forall i=1,\ldots,K \text{ and } \sum_{i=1}^{K}c_i =1 \right\}.\end{equation}
This yields the joint predictive distribution
\begin{equation}
p_t(r_{t+1}|\mathcal{D}) := p(r_{t+1}|R_t,\mathcal{H}_t,\mathcal{D})= \int \left(c_t \otimes I\right)' p\left(r^\text{vec}_{\mathcal{D}, t+1}|R_t,\mathcal{H}_t,\mathcal{D}\right)d r^\text{vec}_{\mathcal{D}, t+1},
\end{equation}
corresponding to a mixture distribution with time-varying weights.
Depending on the choice of the combination weights $c_t$, the scheme balances how much the investment decision is driven by each of the individual models. 
Well-known approaches to combine different models are, among many others, Bayesian model averaging  \citep{Hoeting.1999}, decision-based model combinations \citep{Billio.2013} and optimal prediction pooling \citep{Geweke.2011b}.
In line with the latter, we focus on evaluating the 
goodness-of-fit of the predictive distributions as a measure of predictive accuracy based on rolling-window maximization of the predictive log score,
\begin{equation}\label{equ:optimal_c_logpredictivedistribution}
c^\text{*}_t=\argmax_{c \in \Delta_{[0,1]^K}}\sum\limits_{i=t-{h_c}}^t\log \left[\sum\limits_{k=1}^{K} c_{k} p\left(r_i |R_{i-1},\mathcal{H}_{i-1}, \mathcal{M}_k \right)\right], 
\end{equation}
where $h_c$ is the window size and $p\left(r_i |R_{i-1},\mathcal{H}_{i-1}, \mathcal{M}_k \right)$ is defined as equation in  \eqref{equ:posteriorpredictive}.\footnote{In our empirical application, we set $h_c=250$ days. } If the predictive density concentrates around the observed return values, the predictive likelihood is higher.\footnote{Alternatively, we implemented utility-based model combination in the spirit of \cite{Billio.2013} by choosing $c^*_t$ as a function of past portfolio-performances net of transaction costs. However, the combination scheme resulted in very instable combination weights, putting mas on corner solutions. We therefore refrain from reporting results.}
In general, the objective function of the portfolio optimization problem (EU) is not available in closed form.
Furthermore, the posterior predictive distribution may not arise from a well-known class of probability distributions. Therefore, the computation of portfolio weights depends on (Bayesian) computational methods. We employ MCMC methods to generate draws from the posterior predictive distribution $r_{\mathcal{M}_k,t+1}^{(j)}$ and approximate the integral in optimization problem (EU) by means of Monte-Carlo techniques to compute an optimal portfolio $\hat{\omega}_{t+1}^{\mathcal{M}_k}$.
A detailed description on the computational implementation of this approach is provided in Section A.5 of the Online Appendix.

\section{Predictive Models} \label{sec:models}

As predictive models we choose representatives of three major model classes. First, we include covariance forecasts based on \emph{high-frequency} data utilizing blocked realized kernels as proposed by \cite{Hautsch.2012}. Second, we employ predictions based on parametric models for $\Sigma_t$ using \emph{daily} data. An approach which is sufficiently flexible, while guaranteeing well-conditioned covariance forecasts, is a stochastic volatility factor model according to \cite{Chib.2006}.\footnote{Thanks to the development of numerically efficient simulation techniques by \cite{Kastner.2017},  (MCMC-based) estimation is tractable even in high dimensions. This makes the model becoming one of the very few parametric models (with sufficient flexibility) which are feasible for data of these dimensions.} 
Third, as a candidate representing the class of shrinkage estimators, we employ an approach based on \cite{Ledoit.2003, Ledoit.2004}.

The choice of models is moreover driven by computational tractability in a large-dimensional setting requiring numerical integration through MCMC techniques and in addition a portfolio bootstrap procedure as illustrated in Section \ref{sec:empirics.data}. We nevertheless believe that these models yield major empirical insights, which can be easily transfered to modified or extended approaches.

\subsection{A Wishart Model for Blocked Realized Kernels} \label{sec:BRK}
Realized measures of volatility based on HF data have been shown to provide accurate estimates of daily volatilities and covariances.\footnote{See, e.g., \cite{Andersen.1998}, \cite{Andersen.2003} and \cite{BarndorffNielsen.2009}, among others.}
To produce forecasts of covariances based on HF data, we employ blocked realized kernel (BRK) estimates as proposed by \cite{Hautsch.2012} to estimate the quadratic variation of the price process based on irregularly spaced and noisy price observations. 

The major idea is to estimate the covariance matrix block-wise. 
Stocks are separated into 4 equal-sized groups according to their average number of daily mid-quote observations. The resulting covariance matrix is then decomposed into $b=10$ blocks representing pair-wise correlations within each group and across groups.\footnote{\cite{Hautsch.2015} find that 4 liquidity groups constitutes a reasonable (data-driven) choice for a similar data set. We implemented the setting for up to 10 groups and find similar results in the given framework. }
We denote the set of indexes of the assets associated with block $b \in 1,\ldots,10$ by $\mathcal{I}_b$. For each asset $i$, $\tau_{t,l}^{(i)}$ denotes the time stamp of mid-quote $l$ on day $t$. 
The so-called refresh time is the time it takes for all assets in one block to observe at least one mid-quote update and is formally defined as 
\begin{align}r\tau ^b _{t,1}:=\max_{i \in \mathcal{I}_b}\left\{\tau_{t,1} ^{(i)}\right\}, \hspace{2 cm}r\tau ^b _{t,l+1}:=\max_{i \in \mathcal{I}_b}\left\{\tau_{t,N^{(i)}(r\tau ^b _{t,l})} ^{(i)}\right\},\end{align}
where $N^{(i)} (\tau)$ denotes the number of midquote observations of asset $i$ before time $\tau$.
Hence, refresh time sampling synchronizes the data in time with $r\tau^b _{t,l}$ denoting the time, where all of the assets belonging to block $b$ have been traded at least once since the last refresh time $r\tau^b _{t,l-1}$.
Synchronized log returns are then obtained as $\tilde{r}^{(i)}_{t,l}:=\tilde p^{(i)} _{r\tau^b _{t,l}}-\tilde p^{(i)} _{r\tau^b _{t,l-1}}$, with $\tilde p^{(i)} _{r\tau^{b}_{t,l}}$ denoting the log mid-quote of asset $i$ at time $r\tau^{b}_{t,l}$. 

Refresh-time-synchronized returns build the basis for the multivariate realized kernel estimator by \cite{BarndorffNielsen.2011}, which allows (under a set of assumptions) to consistently estimate the quadratic covariation of an underlying multivariate Brownian semi-martingale price process which is observed under noise. Applying the multivariate realized kernel on each block of the covariance matrix, we obtain 
\begin{align}K_t ^b := \sum\limits_{h=-L_t ^b} ^{L_t ^b} k\left(\frac{h}{L_t ^b +1 }\right) \Gamma_t ^{h,b},\end{align}
where $k( \cdot )$ is the Parzen kernel, $\Gamma_t ^{h,b}$ is the $h$-lag auto-covariance matrix of the assets log returns belonging to block $\mathcal{I}_b$, and $L_t^b$ is a bandwidth parameter, which is optimally chosen according to \cite{BarndorffNielsen.2011}. 
The estimates of the correlations between assets in block $b$ take the form
\begin{align}\hat{H}_t^\text{b} = \left(V_t ^\text{b}\right)^{-1} K_t ^\text{b} \left(V_t ^\text{b}\right)^{-1},\hspace{1 cm}V_t ^\text{b} = \text{diag }\left[K_t ^{b}\right]^{1/2}.\end{align} 
The blocks $\hat{H}_t ^\text{b}$ are then stacked as described in \cite{Hautsch.2012} to obtain the correlation matrix $\hat{H}_t$.
The diagonal elements of the covariance matrix, $\hat\sigma^2_{t,i}$, $i=1,\ldots,N$, are estimated based on univariate realized kernels according to \cite{BarndorffNielsen.2008}. The resulting variance-covariance matrix is then given by
\begin{align}\hat{\Sigma}_t^\text{BRK} = \text{diag} \left(\hat\sigma^2 _{t,1},\ldots,\hat\sigma^2 _{t,N}\right)^{1/2} \hat{H}_t \text{diag} \left(\hat\sigma^2 _{t,1},\ldots,\hat\sigma^2 _{t,N}\right)^{1/2}.\end{align} 
We stabilize the covariance estimates by smoothing over time and computing simple averages of the last $5$ days, i.e., 
$\hat{\Sigma}_{S,t} ^{BRK} := (1/5)\sum_{s=1}^5 \hat{\Sigma}_{t-s+1} ^{BRK}$. 
Without smoothing, HF-data based forecasts would be prone to substantial higher fluctuations, especially on days with extraordinary intra-daily activity such as on the Flash Crash in May 2010. We find that these effects reduce the predictive ability.

We parametrize a suitable return distribution, which is driven by the dynamics of $\hat{\Sigma}_{S,t} ^{BRK}$  and is close in the spirit to \cite{Jin.2013}. 
The dynamics of the predicted return process conditional on the latent covariance $\Sigma_{t}$ are modeled as multivariate Gaussian.
To capture parameter uncertainty, integrated volatility is modeled as an inverse Wishart distribution.\footnote{This represents a multivariate extension of the normal-inverse-gamma approach, applied, for instance, by \cite{Andersson.2001} and \cite{Forsberg.2002}.} 
Thus, the model is defined by:
\begin{align}
\mathcal{L}(r_{t+1}|\Sigma_{t+1}) &\sim N(0,\Sigma_{t+1}), \\
\Sigma_{t+1}|\kappa,B_{t} &\sim W^{-1}_N(\kappa,B_{t}), \\
\kappa B_t &= \hat \Sigma^\text{BRK} _{S,t},\\
\kappa &\sim \exp\left(100\right)\mathbbm{1}_{\kappa>N-1},
\end{align}
where $W^{-1}_N(\kappa,B_{t})$ denotes a $N$-dimensional inverse-Wishart distribution with degrees of freedom $\kappa$ and scale matrix $B_t$, $\exp\left(\lambda\right)$ denotes an exponential distribution with rate $\lambda(=100)$ and $\mathbbm{1}_{\kappa>N-1}$ denotes a truncation at $N-1$.
Though we impose a Gaussian likelihood conditional on $\Sigma_{t+1}$, the posterior predictive distribution of the returns exhibit fat tails after marginalizing out $\Sigma_{t+1}$ due to the choice of the prior. 
\subsection{Stochastic Volatility Factor Models}\label{sec:MVSFV}
Parametric models for return distributions in very high dimensions accommodating time variations in the covariance structure are typically either highly restrictive or computationally (or numerically) not feasible. 
Even dynamic conditional correlation (DCC) models as proposed by \cite{Engle.2002} are not feasible for processes including several hundreds of assets.\footnote{One notable exception is a shrinkage version of the DCC model as recently proposed by \cite{Engle.2017}.} 
Likewise, stochastic volatility (SV) models allow for flexible (factor) structures but have been computationally not feasible for high-dimensional processes either. 
Recent advances in MCMC sampling techniques, however, make it possible to estimate stochastic volatility factor models even in very high dimensions while keeping the numerical burden moderate. 

Employing interweaving schemes to overcome well-known issues of slow convergence and high autocorrelations of MCMC samplers for SV models, \cite{Kastner.2017} propose means to reduce the enormous computational burden for high dimensional estimations of SV objects.
We therefore assume a stochastic volatility factor model in the spirit of \cite{Shephard.1996}, \cite{Jacquier.2002} and \cite{Chib.2006} as given by 
\begin{align}
\tilde{r}_t=\Lambda V(\xi_t)^{1/2}\zeta_t+Q(\xi_t)^{1/2}\varepsilon_t,
\end{align}
where $\Lambda$ is a $N \times j$ matrix of unknown factor loadings, $Q(\xi_t)=\text{diag}\left(\exp(\xi_{1,t}),\ldots,\exp(\xi_{N,t})\right)$ is a $N\times N$ diagonal matrix of $N$ latent factors capturing idiosyncratic effects. The $j \times j$ diagonal matrix $V(\xi_t)=\text{diag}\left(\exp(\xi_{N+1,t}),\ldots,\exp(\xi_{N+j,t})\right)$ captures common latent factors. The innovations $\varepsilon_t \in\mathbb{R}^N$ and $\zeta_t \in\mathbb{R}^j$ are assumed to follow independent standard normal distributions.
The model thus implies that the covariance matrix of $r_t$ is driven by a factor structure
\begin{align}
\text{cov}\left(r_t|\xi_t\right)=\Sigma_t(\xi_t)=\Lambda V_t(\xi_t)\Lambda'+Q_t(\xi_t),
\end{align} 
with $V_t(\xi_t)$ capturing common factors and $Q_t(\xi_t)$ capturing idiosyncratic components.
The covariance elements are thus parametrized in terms of the $N \times j$ unknown parameters, whose dynamics are triggered by j common factors. 
All $N+j$ latent factors are assumed to follow AR(1) processes,
\begin{align}
\xi_{it}=\mu_i+\phi_i(\xi_{t-1, i}-\mu_i)+\sigma_i\eta_{t, i} \text{ } i=1,\ldots,N+j,
\end{align}
where the innovations $\eta_t$ follow independent standard normal distributions and $\xi_{i0}$ is an unknown initial state. 
The AR(1) representation captures the persistence in idiosyncratic volatilities and 
correlations.
The assumption that all elements of the covariance matrix are driven by identical dynamics is obviously restrictive, however, yields parameter parsimony even in high dimensions. Estimation errors can therefore be strongly limited and parameter uncertainty can be straightforwardly captured by choosing appropriate prior distributions for $\mu_i$, $\phi_i$ and $\sigma_i$. 
The approach can be seen as a strong parametric regularization of the covariance matrix which, however, still accommodates important empirical features.
Furthermore, though the joint distribution of the data is conditionally Gaussian,
the stationary distribution exhibits thicker tails. 
The priors for the univariate stochastic volatility processes, $\pi(\mu_i, \phi_i, \sigma_i)$ are in line with \cite{Aguilar.2000}:  
The level $\mu_i$ is equipped with a normal prior, the persistence parameter $\phi_i$ is chosen such that $(\phi_i+1)/2 \sim B(a_0, b_0)$, which enforces stationarity, and for $\sigma^2 _i$ we assume $\sigma^2 _i \sim  G\left(\frac{1}{2}, \frac{1}{2B_\sigma}\right)$.\footnote{In the empirical application we set the prior hyper-parameters to $a_0 = 20, b_0 = 1.5$ and $B_\sigma = 1$ as proposed by \cite{Kastner.2017}. See \cite{Kastner.2016} for further details.} 
For each element of the factor loadings matrix, a hierarchical zero-mean Gaussian distribution is chosen.

\subsection{Covariance Shrinkage}\label{sec:LW}
The most simple and natural covariance estimator is the rolling window sample covariance estimator, 
\begin{align}S_t:=\frac{1}{h-1}\sum\limits_{i=t-h}^t \left(r_i-\hat{\mu}_t\right)\left(r_i-\hat{\mu}_t\right)', \end{align}
with $\hat{\mu}_t := \frac{1}{h-1}\sum_{i=t-h}^t r_i$, and estimation window of length $h$.
It is well-known that $S_t$ is highly inefficient and yields poor asset allocations as long as $h$ does not sufficiently exceed $N$. To overcome this problem,  \cite{Ledoit.2003, Ledoit.2004} propose shrinking $S_t$ towards a more efficient (though biased) estimator of $\Sigma_t$.\footnote{Instead of shrinking the eigenvalues of $S_t$ linearly, an alternative approach would be the implementation of non-parametric shrinkage in the spirit of \cite{Ledoit.2012}. This is left for future research.} The classical linear shrinkage estimator is given by 
\begin{align}\hat{\Sigma}^\text{t,Shrink} = \hat{\delta}F_t +(1-\hat{\delta})S_t,\end{align} 
where $F_t$ denotes the sample constant correlation matrix and $\hat{\delta}$ minimizes the Frobenius norm between $F_t$ and $S_t$. 
$F_t$ is based on the sample correlations $\hat{\rho}_{ij}:=\frac{s_{ij}}{\sqrt{s_{ii}s_{jj}}}$, where $s_{ij}$ is the $i$-th element of the $j$-th column of the sample covariance matrix $S_t$. 
The average sample correlations are given by $\bar{\rho}:=\frac{2}{(N-1)N}\sum\limits_{i=1}^{N}\sum\limits_{j=i+1}^{N-1}\hat{\rho}_{ij}$ yielding the $ij$-th element of 
$F_t$ as $F_{t, ij}=\bar{\rho}\sqrt{\hat{\rho}_{ii}\hat{\rho}_{jj}}$.
Finally, the resulting predictive return distribution is obtained by assuming  $p_t\left(r_{t+1}|\hat{\Sigma}^{t, Shrink}\right) \sim N \left(0,\hat{\Sigma}^{t, Shrink}\right)$. 
Equivalently, a Gaussian framework is implemented for the sample covariance matrix, $p_t(r_{t+1}|S_t) \sim N (0,S_t)$. 
Hence, parameter uncertainty is only taken into account through the imposed regularization of the sample covariance matrix. 
We refrain from imposing additional prior distributions to study the effect of a pure covariance regularization and to facilitate comparisons with the sample covariance matrix. 

\section{Empirical Analysis} \label{sec:empirics} 
\subsection{Data and General Setup}\label{sec:empirics.data} 
In order to obtain a representative sample of US-stock market listed firms, we select all constituents from the S\&P 500 index, which have been traded during the complete time period starting in June 2007, the earliest date for which corresponding HF-data from the LOBSTER database is available. 
This results in a total dataset containing $N=308$ stocks listed at Nasdaq.\footnote{Exclusively focusing on stocks, which are continuously traded through the entire period is a common proceeding in the literature and implies some survivorship bias and the negligence of younger companies with IPO's after 2007. In our allocation approach, this aspect could be in principle addressed by including \emph{all} stocks from the outset and a priori imposing zero weights to stocks in periods, when they are not (yet) traded.} The data covers the period  from June 2007 to March 2017, corresponding to 2,409 trading days after excluding weekends and holidays. 
Daily returns are computed based on end-of-day prices.\footnote{Returns are computed as log-returns in our analysis. As we work with daily data, however, the difference to simple returns is negligible.} 
All the computations are performed after adjusting for stock splits and dividends. 
We extend our data set by HF-data extracted from the LOBSTER database\footnote{See https://lobsterdata.com.}, which provides tick-level message data for every asset and trading day. Utilizing midquotes amounts to more than $73$ billion observations.

In order to investigate the prediction power and resulting portfolio performance of our models, we sequentially generate forecasts on a daily basis and compute the corresponding paths of portfolio weights. More detailed descriptions of the dataset and the computations are provided in Section A.6 of the Online Appendix.
The distinct steps of the estimation and forecasting procedure are as follows: 
We implement $K=4$ different models as of Section \ref{sec:models}. 
The HF approach is based on the smoothed BRK-Wishart model with 4 groups, the SV model is based on $j=3$ common factors, while forecasts based on the sample covariance matrix $S_t$ and its regularized version $\hat{\Sigma}^\text{t, Shrink}$ are computed using a rolling window size of 500 days.\footnote{The predictive accuracy of the SV model is very similar for values of $j$ between $1$ and $5$, but declines when including more factors.} 
 
Moreover, in line with a wide range of studies utilizing daily data, we refrain from predicting mean returns but assume $\mu_t=0$ in order to avoid excessive estimation uncertainty. 
Therefore, we effectively perform global minimum variance optimization under transaction costs and parameter uncertainty as well as model uncertainty.
We sequentially generate (MCMC-based) samples of the predictive return distributions to compute optimal weights for every model. For the model combinations, we sequentially update the combination weights based on past predictive performance to generate the optimal allocation vector $\hat{\omega}_{t+1}^\mathcal{D}$.

In order to quantify the robustness and statistical significance of our results, we perform a bootstrap analysis by re-iterating the procedure described above  $200$ times for random subsets of $N=250$ stocks out of the total $308$ assets.\footnote{
The forecasting and optimization procedures require substantial computing resources. As the portfolio weights at $t$ depend on the allocation at day $t-1$, parallelization is restricted. Sequentially computing the multivariate realized kernels for every trading day, running the MCMC algorithm, performing numerical integration and optimizing in the high-dimensional asset space for all models requires more than one month computing time on a strong cluster such as the Vienna Scientific Cluster.} 

\subsection{Evaluation of the Predictive Performance}
In a first step, we evaluate the predictive performance of the individual models. Note that this step does not require to compute any portfolio weights.\footnote{A detailed visual analysis of the model's forecasts of the high-dimensional return distribution based on generated samples from the posterior predictive distribution is provided in Section A.7 of the Online Appendix.} 
A popular metric to evaluate the performance of predictive distributions is the log posterior predictive likelihood $\log p(r_t^O |R_{t-1},\mathcal{H}_{t-1}, \mathcal{M}_k )$, where $r_{t}^O$ are the observed returns at day $t$, indicating how much probability mass the predictive distribution assigns to the observed outcomes. 
Table \ref{tab:predlikelihood} gives summary statistics of the (daily) time series of the \emph{out-of-sample} log posterior predictive likelihood $\log p(r_t^O |R_{t-1},\mathcal{H}_{t-1}, \mathcal{M}_k )$ for each model. 

In terms of the mean posterior predictive log-likelihood obtained in our sample, the sample covariance matrix solely is not sufficient to provide accurate forecasts. 
Shrinking the covariance matrix, significantly increases its forecasting performance. 
Both estimators, however, still significantly under-perform the SV and HF model. 
The fact that the SV model performs widely similarly to the HF model is an interesting finding as it utilizes only daily data and thus much less information than the latter. 
This disadvantage, however, seems to be overcompensated by the fact that the model captures daily dynamics in the data and thus straightforwardly produces one-step-ahead forecasts. 
In contrast, the HF model produces accurate estimates of $\Sigma_t$, but does not allow for any projections into the future. Our results show that both the efficiency of estimates $\Sigma_t$ and the incorporation of daily dynamics are obviously crucial for superior out-of-sample predictions. 
The fact that both the SV model and the HF model perform widely similar indicates that the respective advantages and disadvantages of the individual models counterbalance each other.\footnote{We thus expect that an appropriate dynamic forecasting model for vast-dimensional covariances, e.g., in the spirit of \cite{Hansen.2012}, may perform even better in terms of the mean posterior predictive log-likelihood, but may contrariwise induce additional parameter uncertainty. 
Given that it is not straightforward to implement such a model in the given general and high-dimensional setting, we leave such an analysis to future research.} 

The last row of Table \ref{tab:predlikelihood} gives the obtained out-of-sample predictive performance of the model combination approach as discussed in Section \ref{sec:combinations}. Computing combination weights $c^*_t$ as described in \eqref{equ:optimal_c_logpredictivedistribution} and evaluating the predictive density, $
LS_{t+1}^\text{Comb.}:= \log \left(\sum\limits_{k=1}^{K} c_{t,k} p\left(r_{t+1}^O |R_{t},\mathcal{H}_{t}, \mathcal{M}_k \right)\right)
$, reflects a significantly stronger prediction performance through the entire sample. 
Combining both high- and low-frequency based approaches thus increases the predictive accuracy and outperforms all individual models.
Thus, our results confirm also in a high-dimensional setting the benefits of mixing high- and low-frequency data for financial time series, documented among others, by \cite{Ghysels.2006}, \cite{Banbura.2013} and \cite{Halbleib.2016}.

Figure \ref{pic:combination_weights_all_assets} depicts the time series of resulting model combination weights, reflecting the relative past prediction performance of each model at each day. We observe distinct time variations, which are obviously driven by the market environment. The gray shaded area in Figure \ref{pic:combination_weights_all_assets} shows the daily  averages of the estimated assets' volatility, computed using univariate realized kernels. We thus observe that during high-volatility periods, the HF approach produces superior forecasts and has the highest weight. 
This is particularly true during the financial crisis and during more recent periods of market turmoil, where the estimation precision induced by HF data clearly pays off. Conversely, forecasts based on daily stochastic volatility perform considerably strong in more calm periods as in 2013/2014. 
Forecasts implied by the shrinkage estimator have lower but non-negligible weight, and thus significantly contribute to an optimal forecast combination. 
In contrast, predictions based on the sample covariance matrix are negligible and are always dominated by the shrinkage estimator.

Superior predictive accuracy, however, is not equivalent to superior portfolio performance. 
An investor \textit{not} adjusting for transaction costs {ex ante} may rebalance her portfolio unnecessarily often if she relies on quickly adjusting predictive models such as SV and -- to less extent -- on the HF-based forecasts.
Though the underlying predictions may be quite accurate, transaction costs can easily offset this advantage compared to the use of the rather smooth predictions implied by (regularized) forecasts based on rolling windows. This aspect is analyzed in the following section. 

\subsection{Evaluation of Portfolio Performance}
To evaluate the performance of the resulting portfolios implied by the individual models' forecasts, we compute portfolio performances based on bootstrapped portfolio weights for the underlying asset universe consisting of 250 assets which are randomly drawn out of the entire asset space.

Our setup represents an investor using the available information to sequentially update her beliefs about the parameters and state variables of the return distribution of the 250 selected assets. 
Based on the estimates, she generates predictions of the returns of tomorrow and accordingly allocates her wealth by solving \eqref{al:EU}, using risk aversion $\gamma= 4$. 
All computations are based on an investor with power utility function $U_\gamma(r) = \frac{\left(1+r\right)^{1-\gamma}}{1-\gamma}$.\footnote{Whereas the optimization framework (EU) does not generally depend on the power utility function, Bayesian numerical methods (as MCMC) allow us to work with almost arbitrary utility functions.
We feel, however, that the specific choice of a power utility function is not critical for the overall quality of our results: First, as advocated, among others, by \cite{Jondeau.2006} and \cite{Harvey.2010}, power utility implies that decision-making is unaffected by scale, as the class of iso-elastic utility function exhibits constant relative risk aversion (CRRA).	
Furthermore, in contrast to quadratic utility, power utility is affected by higher order moments of the return distribution, see also \cite{Holt.2002}.
Thus, power utility allows us to study the effect of parameter- and model uncertainty and makes our results more comparable to the existing literature.}

After holding the assets for a day, she realizes the gains and losses, updates the posterior distribution and recomputes optimal portfolio weights. 
This procedure is repeated for each period and allows analyzing the time series of the realized (''out-of-sample'') returns 
$r^{k} _{t+1}=\sum\limits_{i=1}^N{\hat\omega^{\mathcal{M}_k }}_{t+1, i} r^O _{t+1, i}$. Bootstrapping allows us to investigate the realized portfolio performances of 200 different investors, differing only with respect to the available set of assets.

We assume proportional ($L_1$) transaction costs according to \eqref{al:L1costs}. 
We choose this parametrization, as it is a popular choice in the literature, see, e.g., \cite{DeMiguel.2009b}, and is more realistic than quadratic transaction costs as studied in Section \ref{sec:setup}. 
As suggested by \cite{DeMiguel.2009b}, we fix $\beta$ to $50$bp, corresponding to a rather conservative proxy for transaction costs on the U.S.~market.
 Though such a choice is only a rough approximation to real transaction costs, which in practice depend on (possibly time-varying) institutional rules and the liquidity supply in the market, we do not expect that our general findings are specifically driven by this choice. While it is unavoidable that transaction costs are underestimated or overestimated in individual cases,
we expect that the overall effects of turnover penalization can be still captured with realistic magnitudes. 

The returns resulting from evaluating realized performances net of transaction costs are then given by
\begin{align}
r^{k,\text{nTC}} _{t+1}=r^{k}_{t+1} - \nu_{L_1}\left(\hat{\omega}^{\mathcal{M}_k}_{t+1}\right) = \sum\limits_{i=1}^N{\hat\omega^{\mathcal{M}_k }}_{t+1, i} r^O _{t+1, i} - \nu_{L_1}\left(\hat{\omega}^{\mathcal{M}_k}_{t+1}\right), \quad t = 1,\ldots,T.
\end{align}
We quantify the portfolio performance based on the average portfolio return, its volatility, its Sharpe ratio and the certainty equivalent for $\gamma=4$:
\begin{align}
\hat\mu\left(r^{k,\text{nTC}}\right):= &\frac{1}{T}\sum\limits_{t=1}^{T} r^{k,\text{nTC}}_{t}, \\ 
\hat\sigma\left(r^{k,\text{nTC}}\right):=&\sqrt{\frac{1}{T-1} \sum\limits_{t=1}^{T} \left(r^{k,\text{nTC}}_{t} - \hat\mu\left(r^{k,\text{nTC}}\right) \right)^2}, \\
SR^k :=&\frac{\hat\mu\left(r^{k,\text{nTC}}\right)}{\hat\sigma\left(r^{k,\text{nTC}}\right)}, \\ 
CE^k := &100\left(\left(\frac{1}{T}\sum\limits_{t=1}^T \left(1 + r_t^{k,\text{nTC}}\right)^{1-4}  \right)^{\frac{1}{1-4}}- 1\right).
\end{align}
Moreover, we quantify the portfolio turnover, the average weight concentration, and the average size of the short positions:
\begin{align}
TO^k &:= \frac{1}{T-1} \sum\limits_{t=2}^T \norm{
\hat\omega^{\mathcal{M}_k}_{t} - \frac{\hat\omega^{\mathcal{M}_k}_{t-1}\circ \left(1+r_{t}\right)}{1 + \hat{\omega}^{\mathcal{M}_k \prime}_{t-1} r_{t}} 
}_1, \\
pc^{k} &:=\frac{1}{T}\sum\limits_{t=1}^{T}\sum\limits_{i=1}^N \left(\hat\omega^{\mathcal{M}_k }_{t, i}\right)^2, \\
sp^{k} &:=\frac{1}{T}\sum\limits_{t=1} ^{T} \sum \limits_{i=1} ^N  |\hat\omega_{t, i}| \mathbbm{1}_{\left\{\hat\omega^{\mathcal{M}_k }_{t, i} < 0\right\}}. \end{align}
We compare the performance of the resulting (optimal) portfolios to those of a number of benchmark portfolios. 
First, we implement the naive portfolio allocation $\omega^\text{Naive}:=\frac{1}{N}\iota$ based on daily and bi-monthly rebalancing.\footnote{Implementing naive schemes based on weekly, monthly or no rebalancing at all, does not qualitatively alter the results.} 
We also include the ''classical'' global minimum variance portfolio based on the Ledoit-Wolf shrinkage estimator, $\omega_{t+1} ^\text{mvp} := \frac{\left(\hat\Sigma^\text{t, Shrink}\right)^{-1}\iota}{\iota^\prime\left(\hat\Sigma^\text{t, Shrink}\right)^{-1}\iota}, $
and the optimal global minimum variance weights with a no-short sale constraint, computed as solution to the optimization problem:
\begin{equation}\omega^\text{mvp, no s.} _{t+1} = \argmin \omega'\hat\Sigma^\text{t, Shrink}\omega\text{ s.t. } \iota'\omega = 1 \text{ and }\omega_i \geq 0 \text{ }\forall i=1,\ldots, N. \end{equation}
Furthermore, we include portfolio weights $\omega_{t+1} ^{\vartheta}$ computed as the solution to the optimization problem \eqref{equ:fanregularization} for gross-exposure constraints $\vartheta \in\{1,\ldots,8\}$ as proposed by \cite{Fan.2012}.

Finally, we compare the results to optimal combinations of portfolios which have been proposed in literature: 
Here, the idea of shrinkage is applied \textit{directly} to the portfolio weights by constructing portfolios which optimally (usually with respect to expected utility) balance between the estimated portfolio and a specified target portfolio.
Candidates from this field are a mixture of the efficient portfolio and the naive allocation according to \cite{Tu.2011}, with allocations denoted by $\omega^{tz}$, and mixtures of the efficient portfolio and the minimum variance allocation according to \cite{Kan.2007}, with allocations denoted by $\omega^{kz}$. 
Out of strategies considering estimation risk of the parameters specifically, we utilize the approach of \cite{Jorion.1986}. Here, shrinkage is applied to the estimated moments of the return distribution via traditional Bayesian estimation methods, with the resulting estimates used as inputs to compute the efficient portfolio ($\omega^{bs}$). 
Finally, we create a buy-and-hold portfolio consisting of an investment in all 308 assets, weighted by the market capitalization as of March 2007. This benchmark strategy aims at replicating the market portfolio.
Detailed descriptions of the implementation of the individual strategies are provided in Section A.8 of the Online Appendix.

Table \ref{tab:resultssinglehorserace} summarizes the (annualized) results. 
The results in the first panel correspond to strategies ignoring transaction costs in the portfolio optimization by setting $\nu_{L_1}(\omega)=0$ when computing optimal weights $\omega^*$. These strategies employ the information conveyed by the predictive models, but ignore transaction costs.  However, after rebalancing, the resulting portfolio returns are computed \textit{net} of transaction costs. This corresponds to common proceeding, where turnover penalization is done \textit{ex post}, but is not incorporated in the optimization process.

As indicated by highly negative average portfolio returns, a priori turnover penalizing is crucial in order to obtain a reasonable portfolio performance in a setup with transaction costs. 
The major reason is not necessarily the use of a sub-optimal weight vector, but rather the fact that these portfolio positions suffer from extreme turnover due to a high responsiveness of the underlying covariance predictions to changing market conditions, thus implying frequent rebalancing.
All four predictive models generate average annualized portfolio turnover of more than 50\%, which sums up to substantial losses during the trading period. If an investor would have started trading with 100 USD in June 2007 using the HF-based forecasts without adjusting for transaction costs, she would end up with less than 0.01 USD in early 2017.
None of the four approaches is able to outperform the naive portfolio though the individual predictive models clearly convey information. 
We conclude that the adverse effect of high turnover becomes particularly strong in case of large-dimensional portfolios.\footnote{\cite{Bollerslev.2016} find reverse results for GMV portfolio optimization based on HF-based covariance forecasts. For $N=10$ assets, they find an over-performance of HF-based predictions even in the presence of transaction costs. Two reasons may explain the different findings: First, the burden of a high dimensionality implies very different challenges when working with more than 300 assets. In addition, \cite{Bollerslev.2016} employ methods, which directly impose a certain regularization. In light of our findings in Section \ref{sec:setup}, this can be interpreted as making the portfolios (partly) robust to transaction costs even if the latter are not explicitly taken into account.}

Explicitly adjusting for transaction costs, however, strongly changes the picture: The implemented strategies produce significantly positive average portfolio returns and reasonable Sharpe ratios. The turnover is strongly reduced and amounts to less than 1\% of the turnover implied by strategies which do not account for transaction costs. All strategies clearly outperform the competing strategies in terms of net-of-transaction-cost Sharpe ratios and certainty equivalents.  

Comparing the performance of individual models, the HF-based model, the SV model and the shrinkage estimator yield the strongest portfolio performance, but perform very similarly in terms of $SR$ and $CE$. 
In line with the theoretical findings in Section \ref{sec:setup}, we thus observe that turnover regularization reduces the performance differences between individual models. 
\emph{Combining} forecasts, however, outperforms individual models and yields higher Sharpe ratios. 
We conclude that the combination of predictive distributions resulting from HF data with those resulting from low-frequency (i.e., daily) data is beneficial -- even after accounting for transaction costs. 
Not surprisingly, the sample covariance performs worse but still provides a reasonable Sharpe ratio. 
This confirms the findings in Section \ref{sec:setup} and illustrates that under turnover regularization even the sample covariance yields sensible inputs for high-dimensional predictive return distributions.

Imposing restrictions to reduce the effect of estimation uncertainty, such as a no-short sale constraint in GMV optimization ($\omega^\text{mvp no s.}$), however, does not yield a competing performance. 
These findings underline our conclusions drawn from Proposition \ref{proposition:l1costs}: Though gross-exposure constraints are closely related to a penalization of transaction costs and minimize the effect of misspecified elements in the covariance matrix (see, e.g., \cite{Fan.2012}), such a regularization yields sub-optimal weights in the actual presence of transaction costs.

The last column of Table \ref{tab:resultssinglehorserace} gives the average fraction of days where the portfolio is rebalanced.\footnote{Due to numerical instabilities, we interpret a turnover of less than $0.001$\% as no rebalancing. } The incorporation of transaction costs implies that investors rebalance in just roughly 15\% of all trading days.  
Turnover penalization can be thus interpreted as a Lasso-type mechanism, which enforces buy-and-hold strategies over longer periods as long as markets do not reveal substantial shifts requiring rebalancing. 
In contrast, the benchmark strategies that do not account for transaction costs tend to rebalance permanently (small fractions of) wealth, thus cumulating an extraordinarily high amount of transaction costs. 

To interpret the economic significance of the out-of-sample portfolio performance, we evaluate the resulting utility gains using the framework by \cite{Fleming.2003}. We thus compute the hypothetical fees, an investor with power utility and a relative risk aversion  $\gamma=4$ would be willing to pay on an annual basis to switch from an individual strategy $\mathcal{M}_1$ to strategy $\mathcal{M}_2$.\footnote{We also used alternative values of $\gamma$ with $\gamma =1,\ldots,10$, however, do not find qualitative differences.} The fee is computed such that the investor would be indifferent between the two strategies in terms of the resulting utility. For $r_t ^{\mathcal{M}_1, nTC}$ and $r_t ^{\mathcal{M}_2, nTC}$, we thus determine $\Delta_{\mathcal{M}_1,\mathcal{M}_2}$ such that
\begin{align}
\sum\limits_{t=1}^{T} \left(1 + r_t ^{\mathcal{M}_1, nTC}\right)^{(1-\gamma)} = \sum\limits_{t=1}^{T} \left(1 + r_t ^{\mathcal{M}_2, nTC} - \Delta_{\mathcal{M}_1,\mathcal{M}_2}\right)^{(1-\gamma)}. 
\end{align}
Table \ref{tab:performancefees} shows the average fees an investor would be willing to pay to switch from strategy $\mathcal{M}_i$ in column $i$ to strategy $\mathcal{M}_j$ in row $j$.
We thus find that on average, investors would be willing to pay positive amounts to switch to the superior mixing models. 
For none of the 6 different strategies, the 5\% quantiles of the performance fees are below zero, indicating the strong performance of the mixing model. 
Hence, even after transaction costs investors would gain higher utility by combining forecasts based on high- and low-frequency data.
In our high-dimensional setup, an investor would be willing to pay 7.5 basis points on average on an annual basis to switch from the naive allocation to the leading model combination approach. 
The $1/N$ allocation is thus not only statistically but also economically inferior. 
However, as expected based on the findings in Section \ref{sec:setup}, the switching fees between the individual penalized models are substantially lower. 
This result again confirms the finding that relative performance differences between high-frequency-based and low-frequency-based approaches level out when transaction costs are accounted for.

\section{Conclusions}\label{sec:conclusion}

This paper theoretically and empirically studies the role of transaction costs in large-scale portfolio optimization problems. 
We show that the \textit{ex ante} incorporation of transaction costs regularizes the underlying covariance matrix. 
In particular, our theoretical framework shows the close relation between Lasso- (Ridge-) penalization of turnover in the presence of proportional (quadratic) transaction costs. 
The implied turnover penalization improves portfolio allocations in terms of Sharpe ratios and utility gains. 
This is on the one hand due to regularization effects improving the stability and the conditioning of covariance matrix estimates and on the other hand due to a significant reduction of the amount (and frequency) of rebalancing and thus turnover costs.

In contrast to a pure (statistical) regularization of the covariance matrix, in case of turnover penalization, the regularization parameter is naturally given by the transaction costs prevailing in the market. 
This a priori institution-implied regularization reduces the need for exclusive covariance regularizations (as, e.g., implied by shrinkage), but does not make it superfluous. 
The reason is that a transaction cost regularization only partly contributes to a better conditioning of covariance estimates, but does not guarantee this effect at first place. 
Accordingly, procedures which additionally stabilize predictions of the covariance matrix and their inverse, are still effective but are less crucial as in the case where transaction costs are neglected. 

Performing an extensive study utilizing Nasdaq data of more than 300 assets from 2007 to 2017 and employing more than 70 billion intra-daily trading message observations, we empirically show the following results:
First, we find that neither high-frequency-based nor low-frequency-based predictive distributions result into positive Sharpe ratios when transaction costs are \emph{not} taken into account ex ante.  
Second, as soon as transaction costs are incorporated ex ante into the optimization problem, resulting portfolio performances strongly improve, but differences in the relative performance of competing predictive models become smaller. 
In particular, a portfolio bootstrap reveals that predictions implied by HF-based (blocked) realized covariance kernels, by a daily factor SV model and by a rolling-window shrinkage approach perform statistically and economically widely similarly. 
Third, despite of a similar performance of individual predictive models, mixing HF and low-frequency information is beneficial as it exploits the time-varying nature of the individual model's predictive ability. 
Finally, we find that strategies that \emph{ex-ante} account for transaction costs significantly outperform a large range of well-known competing strategies in the presence of transaction costs and many assets.

Our paper thus shows that transaction costs play a substantial role for portfolio allocation and reduce the benefits of individual predictive models. Nevertheless, it pays off to optimally combine high-frequency and low-frequency information. Allocations generated by adaptive mixtures dominate strategies from individual models as well as any naive strategies.

\section{Acknowledgements}
We thank Lan Zhang, two anonymous referees, Gregor Kastner, Kevin Sheppard, Allan Timmermann, Viktor Todorov and participants of the Vienna-Copenhagen Conference on Financial Econometrics, 2017, the 3rd Vienna Workshop on High-Dimensional Time Series in Macroeconomics and Finance, the Conference on Big Data in Predictive Dynamic Econometric Modeling, Pennsylvania, the Conference on Stochastic Dynamical Models in Mathematical Finance, Econometrics, and Actuarial Sciences, Lausanne, 2017, the 10th annual SoFiE conference, New York, the FMA European Conference, Lisbon, the 70th European Meeting of the Econometric Society, Lisbon, the 4th annual conference of the International Association for Applied Econometrics, Sapporo, the Annual Conference 2017 of the German Economic Association, Vienna, the 11th International Conference on Computational and Financial Econometrics, London, the RIDE Seminar at Royal Holloway, London, the 
Nuffield Econometrics/INET Seminar, Oxford, the research seminar at University Cologne, the 1st International Workshop on “New Frontiers in Financial Markets”, Madrid, and the Brown Bag seminar at Vienna University of Business and Economics for valuable feedback.
We greatly acknowledge the use of computing resources by the Vienna Scientific Cluster. 
Voigt gratefully acknowledges financial support from the Austrian Science Fund (FWF project number DK W 1001-G16) and the IAAE.


\singlespacing
\bibliography{bib}
\newpage
\appendix

\begin{table}[ht]
	\centering
	\begin{tabularx}{0.6\textwidth}{lYYY}
		\toprule
		& Mean                 & SD                     & Best                 \\ \midrule
		HF          & 707.0                & 117.08                 & 0.171                \\
		\addlinespace[-0.7em] & {\tiny[702.9,711.0]} & {\tiny[113.58,119.81]} & {\tiny[0.160,0.181]} \\
		Sample                & 586.5                & 337.56                 & 0.005                \\
		\addlinespace[-0.7em] & {\tiny[580.7,591.8]} & {\tiny[318.99,353.23]} & {\tiny[0.003,0.007]} \\
		LW                    & 694.2                & 165.72                 & 0.197                \\
		\addlinespace[-0.7em] & {\tiny[689.1,699.1]} & {\tiny[161.29,170.00]} & {\tiny[0.185,0.211]} \\
		SV                    & 705.0                & 138.30                 & 0.135                \\
		\addlinespace[-0.7em] & {\tiny[701.6,708.2]} & {\tiny[136.11,140.10]} & {\tiny[0.126,0.146]} \\
		Combination           & {731.3}       & {97.95}         & {0.484}       \\
		\addlinespace[-0.7em] & {\tiny[727.1,735.6]} & {\tiny[96.71,99.34]}   & {\tiny[0.477,0.491]} \\ \bottomrule
	\end{tabularx}
	\caption{Predictive accuracy of the models, evaluated using the time series of out-of-sample predictive log likelihoods, corresponding to $\log p_{t}(r_{t+1}^O|\mathcal{M}_k)$. The values in brackets correspond to bootstrapped 95 \% confidence intervals.  
		\textit{Mean} denotes the average posterior predictive log-likelihood. \textit{SD} is the standard deviation of the time series and \textit{Best} is the fraction of total days where the individual model obtained the highest predictive accuracy among its competitors.
		\textit{Combination} corresponds to predictions created by combining the individual models based on  \eqref{equ:optimal_c_logpredictivedistribution}.}
	\label{tab:predlikelihood}
\end{table}

\begin{table}[ht]
	\centering
	  \resizebox{\textwidth}{!}{%
	\begin{tabular}{lcccccccc}
		\toprule
		&          $\hat\mu$          &        $\hat\sigma$        &         $SR$          &         $CE$          &            $TO$             &             $pc$             &         $sp$          &           \% Trade           \\ \midrule
		\multirow{2}{*}{	$\omega^\text{SV, no TC}$}                        &            -58.9            &           16.12            &  \multirow{2}{*}{-}   &  \multirow{2}{*}{-}   &            55.73            &            31.102            &         2.41          & \multirow{2}{*}{ $>$ 99.95 } \\
		\addlinespace[-0.7em]                                              &    {\tiny[-62.7,-55.8]}     &    {\tiny[15.72,16.52]}    &                       &                       &    {\tiny[53.52,58.16]}     &    {\tiny[30.279,31.942]}    &  {\tiny[2.12,2.78]}   &                              \\
		\multirow{2}{*}{	$\omega^\text{HF, no TC}$}                        &  \multirow{2}{*}{ $<$ -99}  &           15.97            &  \multirow{2}{*}{-}   &  \multirow{2}{*}{-}   &            245.4            &            56.791            &         4.01          & \multirow{2}{*}{ $>$ 99.95 } \\
		\addlinespace[-0.7em]                                              &                             &    {\tiny[15.52,16.33]}    &                       &                       &   {\tiny[239.29,250.43]}    &     {\tiny[55.7,57.892]}     &  {\tiny[3.83,4.65]}   &                              \\ \midrule
		\multirow{2}{*}{	$\omega^\text{Sample}$  }                         &             5.4             &           16.50            &         0.326         &         3.993         &            0.19             &            42.427            &         1.21          &            18.45             \\
		\addlinespace[-0.7em]                                              &      {\tiny[3.8,7.2]}       &    {\tiny[15.87,17.16]}    & {\tiny[0.224,0.446]}  & {\tiny[2.367,5.908]}  &     {\tiny[0.18,0.20]}      &    {\tiny[37.75,47.465]}     &  {\tiny[0.98,1.38]}   &     {\tiny[18.20,19.74]}     \\
		\multirow{2}{*}{	$\omega^\text{LW}$}                               &             5.6             &           14.72            &         0.380          &         4.501         &            0.88             &            24.405            &         0.98          &            11.95             \\
		\addlinespace[-0.7em]                                              &      {\tiny[4.5,6.6]}       &    {\tiny[14.28,15.06]}    & {\tiny[0.307,0.449]}  & {\tiny[3.412,5.497]}  &     {\tiny[0.83,0.93]}      &    {\tiny[22.784,26.427]}    &  {\tiny[0.65,1.11]}   &     {\tiny[10.92,12.96]}     \\
		\multirow{2}{*}{	$\omega^\text{HF}$  }                             &             6.5             &           16.71            &         0.388         &         5.081         &            0.14             &            17.165            &         0.48          &            14.44             \\
		\addlinespace[-0.7em]                                              &      {\tiny[5.9,6.9]}       &    {\tiny[16.52,16.93]}    & {\tiny[0.354,0.414]}  & {\tiny[4.488,5.518]}  &     {\tiny[0.09,0.21]}      &    {\tiny[15.937,18.062]}    &  {\tiny[0.21,0.61]}   &     {\tiny[13.41,15.54]}     \\
		\multirow{2}{*}{	$\omega^\text{SV}$  }                             &             6.6             &           16.76            &         0.391         &         5.153         &            0.32             &            17.860             &         0.71          &            14.13             \\
		\addlinespace[-0.7em]                                              &       {\tiny[5.9,7.0]}        &    {\tiny[16.59,16.99]}    & {\tiny[0.354,0.417]}  & {\tiny[4.489,5.589]}  &     {\tiny[0.21,0.41]}      &    {\tiny[16.086,18.827]}    &  {\tiny[0.68,0.92]}   &     {\tiny[13.11,15.20]}     \\
		\multirow{2}{*}{	$\omega^\text{Comb. }$}                           &             6.5             &           14.86            &         0.438         &         5.401         &            0.49             &            20.609            &         0.81          &            13.70             \\
		\addlinespace[-0.7em]                                              &       {\tiny[5.0,7.8]}        &    {\tiny[14.45,15.39]}    &  {\tiny[0.342,0.52]}  & {\tiny[3.877,6.693]}  &     {\tiny[0.46,0.53]}      &    {\tiny[18.483,23.115]}    &  {\tiny[0.54,0.92]}   &     {\tiny[12.53,14.83]}     \\ \midrule
		\multirow{2}{*}{	$\omega^\text{Naive}$ }                           &             5.2             &            23.5            &         0.224         &         2.399         &            1.10             &             0.4              & \multirow{2}{*}{ - }  & \multirow{2}{*}{ $>$ 99.95 } \\
		\addlinespace[-0.7em]                                              &      {\tiny[4.9,5.5]}       &    {\tiny[23.14,23.84]}    & {\tiny[0.204,0.234]}  & {\tiny[2.035,2.733]}  &     {\tiny[1.09,1.12]}      &       {\tiny[0.4,0.4]}       &                       &                              \\
		\multirow{2}{*}{	$\omega^\text{Naive, 2m}$}                        &             5.9             &           23.03            &         0.258         &         3.294         &            0.15             &            0.403             & \multirow{2}{*}{ - }  &             1.66             \\
		\addlinespace[-0.7em]                                              &      {\tiny[5.7,6.2]}       &    {\tiny[22.68,23.36]}    & {\tiny[0.245,0.273]}  & {\tiny[2.991,3.632]}  &     {\tiny[0.14,0.16]}      &     {\tiny[0.403,0.404]}     &                       &      {\tiny[1.63,1.66]}      \\
		\multirow{2}{*}{	$\omega^\text{mvp}$ }                             &            -59.4            &           14.11            &  \multirow{2}{*}{ -}  &  \multirow{2}{*}{ -}  &            53.49            &            45.65             &         3.33          & \multirow{2}{*}{ $>$ 99.95 } \\
		\addlinespace[-0.7em]                                              &    {\tiny[-62.8,-56.1]}     &    {\tiny[13.73,14.55]}    &                       &                       &    {\tiny[52.25,54.64]}     &    {\tiny[42.864,48.28]}     &  {\tiny[3.24, 3.42]}  &                              \\
		\multirow{2}{*}{	$\omega^\text{mvp, no s.}$}                       &             2.8             &            13.1            &         0.213         &         1.935         &            3.53             &            10.616            & \multirow{2}{*}{ - }  & \multirow{2}{*}{ $>$ 99.95 } \\
		\addlinespace[-0.7em]                                              &       {\tiny[2.0,3.6]}        &   {\tiny[12.97, 13.29]}    & {\tiny[0.149, 0.279]} & {\tiny[1.092, 2.794]} &     {\tiny[3.34, 3.69]}     &    {\tiny[9.161, 11.548]}    &                       &                              \\
		\multirow{2}{*}{	$\omega^\text{kz, 3f}$ } & \multirow{2}{*}{ $<$ -99 }  & \multirow{2}{*}{  $>$ 50 }  &           -           &           -           &            941.9            &  \multirow{2}{*}{ $>1000$ }  &         5.98          & \multirow{2}{*}{ $>$ 99.95 } \\
		\addlinespace[-0.7em]                                              &                             &                            &                       &                       &   {\tiny[175.51,2339.7]}    &                              &  {\tiny[3.72, 8.8]}   &                              \\
		\multirow{2}{*}{	$\omega^\text{tz}$ }                              & \multirow{2}{*}{  $<$ -99 } & \multirow{2}{*}{   $>$ 50 } &           -           &           -           & \multirow{2}{*}{  $>1000$}  &  \multirow{2}{*}{ $>1000$ }  &         28.09         & \multirow{2}{*}{ $>$ 99.95 } \\
		\addlinespace[-0.7em]                                              &                             &                            &                       &                       &                             &                              & {\tiny[17.31, 57.51]} &                              \\
		\multirow{2}{*}{	$\omega^\text{bs}$ }                              & \multirow{2}{*}{ $<$ -99 }  & \multirow{2}{*}{   $>$  50} &           -           &           -           & \multirow{2}{*}{ $> 1000$ } & \multirow{2}{*}{  $>1000$  } &         124.7         & \multirow{2}{*}{ $>$ 99.95 } \\
		\addlinespace[-0.7em]                                              &                             &                            &                       &                       &                             &                              & {\tiny[39.8, 228.14]} &                              \\
		\multirow{2}{*}{	$\omega^{\vartheta}$ }                              & -24.28  & 12.04 &           -           &           -           & 25.846 & 18.692 &         1.499        & \multirow{2}{*}{ $>$ 99.95 } \\
		\addlinespace[-0.7em]                                              &      {\tiny[ -26.28, -22.15]}                         &   {\tiny[11.78, 12.34]}                         &                       &                       &       {\tiny[25.43, 26.31]}                      &             {\tiny[18.05, 19.28]}                 & {\tiny[1.40, 1.58]} &                              \\
		\midrule
\multirow{1}{*}{	$\omega^\text{Market}$} & 6.2 & 21.19 & 0.292 & 3.958 & - & - &  - & - \\
  \bottomrule
	\end{tabular}
}
	\caption{ Annualized averages of the bootstrapped out-of-sample portfolio performances based on $1904$ trading days. Transaction costs are proportional to the $L_1$ norm of re-balancing (as of \eqref{al:L1costs}). $\hat \mu$ is the annualized portfolio return in percent, $\hat \sigma$ is the annualized standard deviation in percent, $SR$ denotes the (annualized) out-of-sample Sharpe ratio of the individual strategies, $CE$ is the Certainty Equivalent for an investor with power utility function and risk-aversion factor $\gamma=4$, $TO$ is the average turnover in percent, $pc$ is the average weight concentration ($L_2$ norm  of the portfolio weights) and $sp$ is the average proportion of the sum of negative portfolio weights. 
		\textit{\% Trade} is the fraction of days with trading activity more than 0.001\%.}  
	\label{tab:resultssinglehorserace}
\end{table}

\begin{table}[ht]
	\centering
	\begin{tabular}{rcccccc}
		\toprule
		&       Naive        &         LW         &       Sample        &         SV          &          HF          &     Combination     \\
		mvp no s. &        1.98        &        6.37        &        6.16         &        6.01         &         6.41         &        7.73         \\
		\addlinespace[-0.7em] & {\tiny[1.86,2.01]} & {\tiny[6.24,7.51]} & {\tiny[5.98,7.56]}  & {\tiny[5.80,6.89]}  &  {\tiny[6.32,7.25]}  & {\tiny[7.12,8.02]}  \\
		Naive &                    &        6.18        &        5.99         &        5.83         &         6.40         &        7.57         \\
		\addlinespace[-0.7em] &                    & {\tiny[5.65,7.86]} & {\tiny[5.41,6.43]}  & {\tiny[5.44,6.46]}  &  {\tiny[6.29,6.82]}  & {\tiny[7.25,7.89]}  \\
		LW &                    &                    &        0.01         &        0.41         &         0.11         &        1.44         \\
		\addlinespace[-0.7em] &                    &                    & {\tiny[-0.29,0.60]} & {\tiny[-0.01,0.91]} & {\tiny[-0.43,0.86]}  & {\tiny[0.57,2.25]}  \\
		Sample &                    &                    &                     &         0.1         &         0.3          &        1.38         \\
		\addlinespace[-0.7em] &                    &                    &                     & {\tiny[-0.22,0.89]} & {\tiny[-0.12, 1.32]} & {\tiny[0.84,2.13]}  \\
		SV &                    &                    &                     &                     &        - 0.68        &        1.11         \\
		\addlinespace[-0.7em] &                    &                    &                     &                     & {\tiny[-0.29, 1.02]} & {\tiny[0.95, 1.81]} \\
		HF &                    &                    &                     &                     &                      &        1.31         \\
		\addlinespace[-0.7em] &                    &                    &                     &                     &                      & {\tiny[0.88,1.79]}  \\ \bottomrule
	\end{tabular}
	\caption{Average annual performance fees in basis points for switching between the individual models based on $\gamma=4$. 
	The table reads: \textit{On an annual basis an investor would be willing to pay 7.57 basis point to switch from the naive portfolio to the combination strategy. 
	The number in brackets denote the bootstrapped 95 \% confidence intervals. }}
	\label{tab:performancefees}
\end{table}

\begin{figure}[ht]
	\resizebox{0.7\textwidth}{!}{\includegraphics{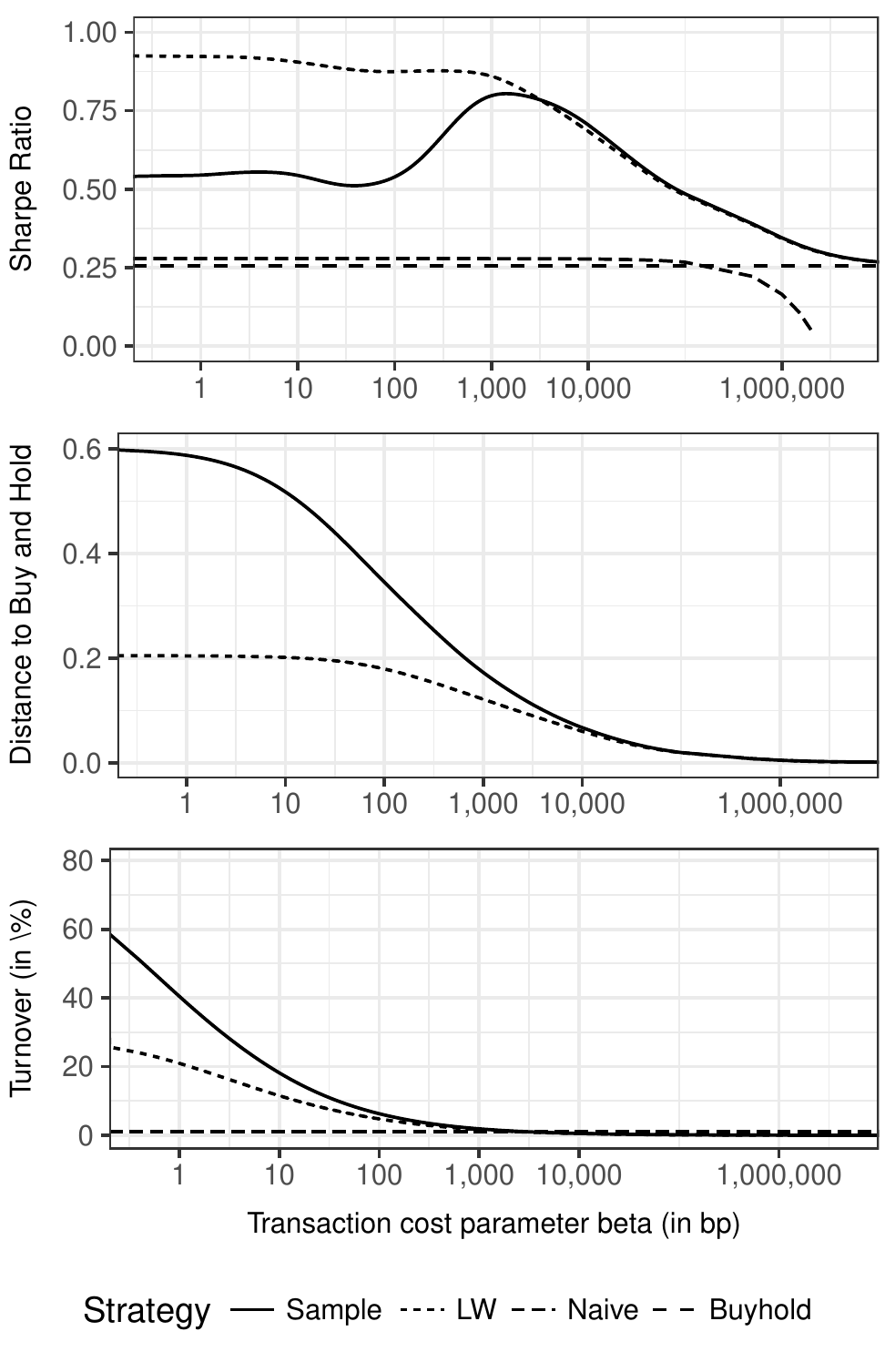}}
	\centering
	\caption{Annualized realized Sharpe ratio net of transaction costs for portfolios based on $N=308$ assets with daily reallocation for the period from June 2007 until March 2017 with quadratic transaction costs. 
		The portfolio weights are computed as solutions to optimization problem \eqref{equ:maximization_naive} with varying choices of the transaction cost parameter $\beta$ (in basis points). 
		The covariance matrix is estimated using the sample covariance matrix $S_t$ and a regularized version $\hat{\Sigma}^{t,\text{Shrink}}$ \citep{Ledoit.2003, Ledoit.2004} (both based on a rolling window of 500 days). 
		The mean of the returns, $\mu_t$, is set to 0. 
		The second panel corresponds to the average $L_1$ difference of the estimated weights to the buy-and-hold portfolio with naive ($1/N$) initial allocation.
		The third panel illustrates average turnover (in percent), measured as $L_1$ difference between $\omega_{t+1}^\beta$ and $\omega_{t^+}^\beta$.
		Strategy \textit{Naive} denotes an allocation with daily rebalancing to the naive allocation, \textit{Buyhold} corresponds to implementing 1/N once and refraining from rebalancing afterwards.}
	\label{fig:sharpe_ratio_varying_beta}
\end{figure}

\begin{figure}[ht]
	\resizebox{0.80\textwidth}{!}{\includegraphics{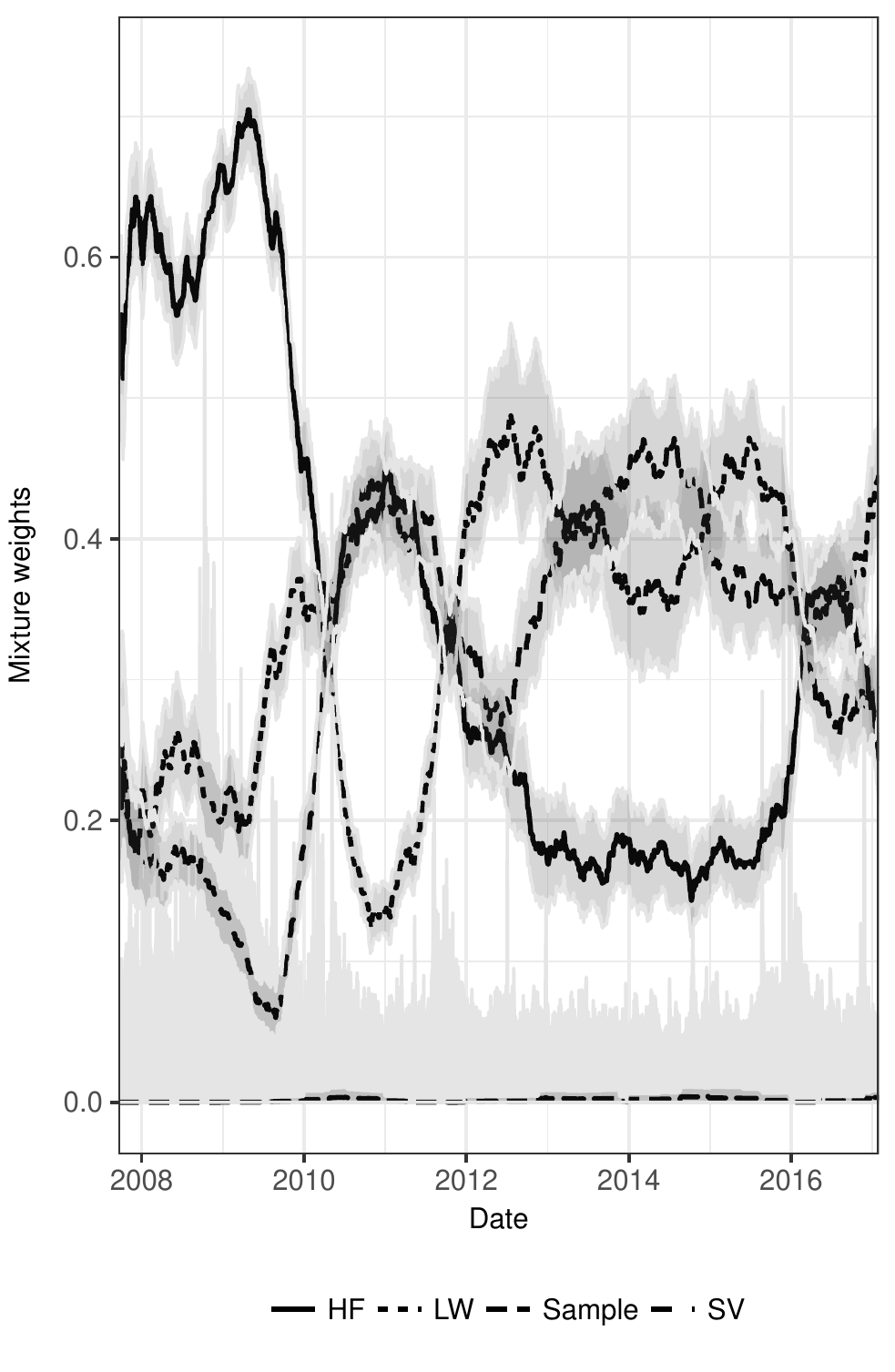}}
	\centering
	\caption{Combination weights computed based on \eqref{equ:optimal_c_logpredictivedistribution}. The weights are sequentially updated to maximize the log posterior predictive likelihood of past observed returns based on a rolling-window of 250 days. The lines correspond to the mean of the bootstrapped combination weights.
		The gray confidence bands denote the 2.5\% and 97.5\% quantiles of the bootstrapped combination weights.  	
		The shaded gray area in the background visualizes a scaled time series of the daily estimated volatility across all $308$ markets.}
	\label{pic:combination_weights_all_assets}
\end{figure}

\end{document}